\documentclass[a4paper,11pt]{article}
\pdfoutput=1
\usepackage{a4wide}
\usepackage[UKenglish]{babel}
\usepackage[noadjust]{cite}
\usepackage{amsmath,amssymb,amsfonts}
\usepackage{amsthm} 
\usepackage{bbm}
\usepackage{mathtools}
\usepackage{mathrsfs}
\usepackage{csquotes}
\usepackage{todonotes}
\usepackage{xcolor}
\usepackage{blindtext}
\usepackage{enumerate}
\usepackage{manfnt}
\usepackage{multicol}
\usepackage{microtype}
\usepackage{xparse}
\usepackage{thm-restate}

\usepackage{hyperref} 
\hypersetup{colorlinks=true,citecolor=blue,linkcolor=blue,urlcolor=blue}
\usepackage[capitalize,nameinlink]{cleveref}
\MakeRobust{\ref}

\usepackage{stmaryrd}
\usepackage{dutchcal}
\usepackage[scr=boondox]{mathalfa}

\usepackage{caption}
\DeclareCaptionFormat{twolinetab}{#1#2\\#3}
\newcommand{\YES}{\textsc{Yes}}

\DeclareSymbolFont{extraup}{U}{zavm}{m}{n}
\DeclareMathSymbol{\varheart}{\mathalpha}{extraup}{86}
\DeclareMathSymbol{\vardiamond}{\mathalpha}{extraup}{87}

\newcommand{\Test}[2]{
\def\temp{#2}\ifx\temp\empty
  \operatorname{Test}_{#1}
\else
  \operatorname{Test}_{#1}^{#2}
\fi
}

\newcommand{\C}{\mathbb{C}}

\newcommand{\K}{\mathbf{K}}

\newcommand{\X}{\mathbf{X}}
\newcommand{\Y}{\mathbf{Y}}

\newcommand{\N}{\mathbb{N}}
\newcommand{\R}{\mathbb{R}}

\newcommand{\HH}{\mathbb{H}}

\newcommand{\NAE}{\mathbf{NAE}}
\renewcommand{\vec}[1]{\mathbf{#1}}
\newcommand{\ba}{\vec{a}}

\newcommand{\bc}{\vec{c}}

\newcommand{\bs}{\vec{s}}

\newcommand{\bx}{\vec{x}}
\newcommand{\bv}{\vec{v}}
\newcommand{\by}{\vec{y}}
\newcommand{\bw}{\vec{w}}
\newcommand{\be}{\vec{e}}

\newcommand{\bz}{\vec{z}}

\DeclareMathAlphabet{\mathbx}{U}{BOONDOX-ds}{m}{n}
\SetMathAlphabet{\mathbx}{bold}{U}{BOONDOX-ds}{b}{n}
\DeclareMathAlphabet{\mathbbx} {U}{BOONDOX-ds}{b}{n}

\DeclareMathOperator{\ptn}{ptn}

\DeclareMathOperator{\tr}{tr}

\DeclareMathOperator{\Pol}{Pol}
\DeclareMathOperator{\PCSP}{PCSP}
\DeclareMathOperator{\CSP}{CSP}

\DeclareMathOperator{\ar}{ar}

\DeclarePairedDelimiterX{\norm}[1]{\lVert}{\rVert}{#1}

\newcommand{\Sminion}{\ensuremath{{\mathscr{S}}}}

\newcommand{\bzero}{\mathbf{0}}

\theoremstyle{plain}
\newtheorem{thm}{Theorem}
\newtheorem*{thm*}{Theorem}
\newtheorem{lem}[thm]{Lemma}
\newtheorem*{lem*}{Lemma}
\newtheorem{prop}[thm]{Proposition}
\newtheorem*{prop*}{Proposition}
\newtheorem{cor}[thm]{Corollary}
\newtheorem{conj}[thm]{Conjecture}

\newtheorem{fact}[thm]{Fact}
\theoremstyle{definition}
\newtheorem{defn}[thm]{Definition}
\newtheorem{rem}[thm]{Remark}

\newtheorem{example}[thm]{Example}

\AtBeginDocument{%
 \providecommand\BibTeX{{%
   \normalfont B\kern-0.5em{\scshape i\kern-0.25em b}\kern-0.8em\TeX}}}

\NewDocumentCommand \core { m }{\operatorCore(#1)}
\DeclareMathOperator{\operatorCore}{core}
\NewDocumentCommand \relstr { m } {\mathbf{#1}}
\NewDocumentCommand \uprelstr { > {\SplitArgument {1} {,}} m } {\inneruprel#1}
\NewDocumentCommand \inneruprel { m m } {{#2}^{\mathbf \rightarrow#1}}
\NewDocumentCommand \tuple { m } {\mathbf{#1}}
\NewDocumentCommand \GameAB { > {\SplitArgument {1} {,}} m } {\internalGameAB#1}
\NewDocumentCommand \internalGameAB { m m } {\operatorGameAB^{\mathbf{\rightarrow}#1}(\relstr{#2})}
\DeclareMathOperator{\operatorGameAB}{CSP}
\NewDocumentCommand \MKCSP {m} {\operatorCSP(#1)}
\DeclareMathOperator{\operatorCSP}{CSP}

\NewDocumentCommand \GameBA { > {\SplitArgument {1} {,}} m } {\internalGameBA#1}
\NewDocumentCommand \internalGameBA { m m } {\operatorGameBA^{\mathbf{\leftarrow} #1}(\relstr{#2})}
\DeclareMathOperator{\operatorGameBA}{CSP}

\NewDocumentCommand \uprelstrbob { > {\SplitArgument {1} {,}} m } {\inneruprelbob#1}
\NewDocumentCommand \inneruprelbob { m m } {{#2}^{\mathbf \leftarrow #1}}

\DeclareMathOperator{\PVM}{PVM}

\usepackage[blocks]{authblk}
\usepackage{thmtools}
\usepackage{xpatch}
\makeatletter
\xpatchcmd{\thmt@restatable}
{\csname #2\@xa\endcsname\ifx\@nx#1\@nx\else[{#1}]\fi}
{\ifthmt@thisistheone\csname #2\@xa\endcsname\ifx\@nx#1\@nx\else[{#1}]\fi\else\csname #2\@xa\endcsname\fi}
{}{} 
\makeatother
\begin{document}

\title{Classical simulation of quantum CSP strategies%
\thanks{Partially funded by the National Science Centre, Poland under the Weave program grant no. 2021/03/Y/ST6/00171 and by UKRI EP/X024431/1. For the purpose of Open Access, the authors have applied a CC BY public copyright licence to any Author Accepted Manuscript (AAM) version arising from this submission.}%
}

\renewcommand*{\Affilfont}{\small\normalfont}
\renewcommand*{\Authands}{, }
\setlength{\affilsep}{2em}

\author[a,b]{Demian Banakh}
\author[c]{Lorenzo Ciardo}
\author[a]{Marcin Kozik}
\author[a,b,d]{Jan Tułowiecki}

\affil[a]{Faculty of Mathematics and Computer Science, Jagiellonian University, Poland
}
\affil[b]{Doctoral School of Exact and Natural Sciences, Jagiellonian University, Poland
}
\affil[c]{Department of Computer Science, University of Oxford, UK
}
\affil[d]{BEIT, Poland}


\date{\vspace{-5ex}}

\maketitle

\begin{abstract}
\noindent We prove that any perfect quantum strategy for the two-prover game encoding a constraint satisfaction problem (CSP) can be simulated via a perfect classical strategy with an extra classical communication channel, whose size depends only on $(i)$ the size of the shared quantum system used in the quantum strategy, and $(ii)$ structural parameters of the CSP template.
The result is obtained via a combinatorial characterisation of perfect classical strategies with extra communication channels and a geometric rounding procedure for the projection-valued measurements involved in quantum strategies. 

A key intermediate step of our proof is to establish that the gap between the
classical chromatic number of graphs and its quantum variant
is bounded when the quantum strategy involves shared quantum information of bounded size.
\end{abstract}

\section{Introduction}
Two-prover games yield a convenient paradigm for investigating the power and limits of computation and information transfer. 
Some of the most striking developments of modern theoretical computer science---for instance, in hardness of approximation (e.g., \cite{bellare1993inapprox2prover,raz1995parallel,bellare1995nonlinearinapprox,feige1996ipinapprox,Khot02stoc,Khot18:focs-pseudorandom}) and in the theory of complexity classes (e.g., \cite{cai1994pspace,babai1991nexp2prover})---are most conveniently formulated as statements on the expressive power of two cooperating provers/players interacting with a verifier/referee.

\sloppy Consider the scenario of a \textit{Referee} challenging two players---\textit{Alice} and \textit{Bob}---by asking each of them a question involving a local set of variables for some fixed predicate.
Alice and Bob aim to convince the Referee that the predicate is true. 
Just like in a real questioning scenario, the replies should be \textit{plausible} (i.e., they should locally satisfy the predicate) and \textit{consistent} with each other.
%
%
In this work, we shall be concerned with \textit{perfect strategies} for the cooperating players:
\begin{center}
    \textit{When can Alice and Bob convince the Referee, no matter which questions the Referee chooses?}
\end{center}
If the players are allowed to communicate freely during the game, they can always win as long as each question of the Referee admits some plausible answer.
At the other extreme, if no communication is allowed, we expect a perfect strategy to exist precisely when the predicate is true. After all, if Alice and Bob \textit{are} innocent, they should simply tell the truth, and their replies will always sound both plausible and consistent; if they are \textit{not} innocent, any strategy will have a weak point, which a zelous Referee might use to spot implausibility or inconsistency in the answers.


The study of two-prover games has also shown to be central in quantum theory---specifically, to shed light on the compatibility of quantum and classical physics.
At a high level, the idea is that, once physics enters the picture, the ``no communication'' assumption between Alice and Bob during the game sets different restrictions on the type of correlations exhibited by Alice and Bob's answers, depending on the physical theory we are adopting. Specifically, in a \textit{quantum} model of physics, the phenomenon of entanglement might in principle be used by the players to devise strategies resulting in answers that are more correlated than those resulting from any strategy admitted in a \textit{classical} theory. More precisely, a \textit{quantum strategy} consists of a bipartite quantum system shared by Alice and Bob, as well as two families $\mathcal{A}$ and $\mathcal{B}$ of  measurements, one for Alice and one for Bob, with different measurements associated with the possible questions the Referee might ask.
%
%
Upon receiving the questions, 
each of the players performs the corresponding measurement from $\mathcal{A}$ and $\mathcal{B}$, respectively, on their part of the shared quantum system.
The outcomes of the measurements determine their answers.
Some predicates admit a quantum strategy that outperforms \textit{any} classical strategy. This has served as theoretical evidence for the non-classical properties of quantum physics---in particular, non-locality and contextuality, see e.g.~\cite{bell1964einstein,clauser1969proposed,mermin1990simple,brassard1999cost,brassard2005quantum_pseudo}. 
\bigskip

The goal of this work is to investigate the extent to which the correlation between Alice's and Bob's answers achieved via quantum entanglement can be simulated by introducing into the game a limited amount of \textit{classical} communication among the cooperating players.
To that end, we consider the scenario where Alice and Bob, upon receiving the questions from the Referee, are allowed to send a classical message one to the other.
Our main contribution is to show that \textit{any perfect quantum strategy can be converted into a perfect classical strategy with a communication channel whose size depends only on $(i)$ the dimension of the Hilbert space describing the shared system in the quantum strategy, and $(ii)$ the type of predicate.}

It shall be convenient, in our analysis, to treat the predicate describing a game as an instance of a \textit{constraint satisfaction problem} $\CSP(\Y)$. Here, $\Y$ is the \textit{template} of the CSP, encoding all plausible replies of the players to the Referee's questions.
In a zero-communication game,
the existence of a perfect classical strategy means that the predicate yields a $\YES$-instance of $\CSP(\Y)$. Moreover, it was shown in~\cite{abramsky2017quantum} that a perfect quantum strategy corresponds to a $\YES$-instance of the CSP parameterised by
a certain infinite-domain, quantised version $\Y^\HH$ of $\Y$.

As the first step of our proof (in~\cref{sec_combinatorial_characterisation_strategies}), we show combinatorially that a similar phenomenon also holds for classical strategies with extra communication channels. Specifically, the two scenarios of Alice messaging Bob and Bob messaging Alice are captured by two distinct modifications of $\Y$, which we denote by $\uprelstr{k,\Y}$ and $\uprelstrbob{k,\Y}$, respectively, with $k$ representing the size of the message. 
As a consequence, classical strategies with an extra communication channel for a $\Y$-predicate are equivalent to classical strategies with no communication for a $\uprelstr{k,\Y}$-predicate and a $\uprelstrbob{k,\Y}$-predicate, respectively. Unlike the structure $\Y^\HH$ for the quantum case, such structures are finite. 

The second step of our proof (in~\cref{sec_entangled_vs_classical}) is geometric, and consists in a rounding procedure that, given a perfect quantum strategy on a $d$-dimensional Hilbert space, produces a classical strategy that is perfect when a message of suitable size (depending on $d$ and on $\Y$, but not on the instance) is allowed.
By virtue of the first step, this rounding procedure can be performed at the template level rather than at the game level, and is thus instance-free. For example, showing that any perfect quantum strategy can be simulated by a perfect classical strategy with an ``Alice to Bob'' message of size $k$ 
amounts to building a relation-preserving map between $\Y^\HH$ and $\uprelstr{k,\Y}$.

The main technical challenge of the rounding procedure is to show that the quantum version of any structure $\Y$ has a finite classical chromatic number. This corresponds to finding a finite colouring for the infinitely many measurements admitted in a quantum strategy.
We give two different proofs of this fact.
One proof works for arbitrary CSPs, and makes use of the compactness of the unitary group in $\C^{d\times d}$ to build a system of indicator frames that ``almost diagonalises'' any projective measurement. The second proof only applies to digraphs, 
and it is based on a classic upper bound on the size of sphere coverings by~\cite{rogers1963covering}. The latter argument allows simulating a perfect quantum strategy via a 
classical communication channel whose size asymptotically matches the dimension of the Hilbert space describing the shared quantum system, in the digraph case. We do not know whether this fact is a coincidence or rather an instance of a deeper correspondence between shared quantum information and classical communication in the context of non-local games.

As a by-product of the second step of our proof, in~\cref{subsec_quantum_vs_classical_chromatic_gap}, we obtain an upper bound on the gap between the classical chromatic number $\chi$ and its $d$-dimensional quantum version $\chi_d$, by showing that $\chi(\textbf{G})\leq \alpha_d\cdot\chi_d(\textbf{G})
$ for each graph $\textbf{G}$, where $\alpha_d$ depends only on the dimension $d$ of the quantum system.

\paragraph{Related work}
Our result finds its collocation within a line of work in quantum information theory (specifically, quantum communication complexity) aiming to quantify the communication cost of classically simulating the correlations observed in quantum theory. This direction was pioneered independently in~\cite{maudlin1992d,brassard1999cost,steiner2000towards}. In particular, it was shown in~\cite{brassard1999cost}
that the correlations produced by the two parties of a Bell experiment through two-outcome projective measurements on a single pair of qubits in a
Bell state can be simulated by
classical strategies augmented with eight bits of communication. Later,~\cite{csirik2002cost} proved that six bits are sufficient, and~\cite{toner2003communication} reduced the number to only one bit, which is optimal; see also~\cite{bacon2003bell}, where the latter result was extended to a complete characterisation of the polytope of the admitted correlations arising from such strategies. Moreover,~\cite{brassard1999cost} gave a lower bound (but no upper bound) on the number of classical communication bits required in the case of a system of arbitrarily many Bell states shared by the players. It was later shown in~\cite{regev2010simulating} that two communication bits are sufficient (and, as proved in~\cite{vertesi2009lower}, necessary) for exactly simulating the quantum correlations  arising in the case that both players perform two-outcome measurements on a shared quantum system described by a Hilbert space of arbitrary dimension. In the latter work, the marginal distributions produced by the classical protocol with communication do not coincide with those coming from the quantum measurements. 
The case of multipartite---as opposed to bipartite---entanglement was considered in~\cite{brassard2014exact}; see
also~\cite{coates2002quantum,branciard2012classical,zambrini2019bell,renner2023minimal,sidajaya2023neural} in the same line of work. 

We note that the number of measurement outcomes in the protocols discussed above corresponds to the size of the CSP template in our setting (see Section~\ref{subsec_quantum_strategies}).
Furthermore,
in this line of work, a classical protocol is typically said to simulate a quantum measurement scenario if the outputs it produces have exactly the same bivariate distribution as those produced by the quantum measurements---or if the correlations coincide, as in~\cite{regev2010simulating}.
The notion of simulation we use in the current work is less restrictive: We only require that the classical distributions be perfect (i.e., satisfy all local constraints of the CSP) when the quantum distributions are perfect. In fact, in order to exactly simulate quantum correlations via classical messages of finite size, Alice and Bob provably need infinite shared randomness---as was shown in~\cite{massar2001classical}, see also~\cite[\S III.C]{brunner2014bell}. We also observe that the literature on classical simulation of quantum correlations often considers measurement scenarios on quantum systems prepared in some specific state (for example, a Bell state), while the current work captures quantum strategies involving a quantum system prepared in an arbitrary state. Nevertheless, it was shown in~\cite[Thm.~5]{abramsky2017quantum} (generalising results from~\cite{CameronMNSW07,roberson2013variations,MancinskaR16}) that, as far as perfect strategies are concerned, there is no loss of generality in preparing the shared quantum system in a maximally entangled state.

Finally, we point out that other ways of quantifying the classical resources needed to simulate non-classical effects of quantum systems have been investigated in the literature. A primary example is the so-called memory cost of contextuality---i.e., the size of a classical system needed
to simulate the certain predictions that can be obtained from sequences of measurements on a
quantum system, see~\cite{kleinmann2011memory,cabello2018optimal}.

\section{Games and strategies}
\label{sec_strategies}
Before embarking on the proof of our main results, we give in this section a formal description of the two-player games associated with CSP instances, as well as the types of strategies that we will be concerned with in this work.

Let $\sigma$ be a (relational) \textit{signature}; i.e., a finite set of symbols $R$, each with an associated positive integer $\ar(R)$ called the \textit{arity} of $R$. A \textit{relational structure} with signature $\sigma$ (in short, a \textit{$\sigma$-structure}) $\Y$ consists of a set $Y$ called \textit{domain} or \textit{universe}, and a \textit{relation} $R^\Y\subseteq Y^{\ar(R)}$ for each symbol $R\in\sigma$. Given two $\sigma$-structures $\X$ and $\Y$, a \textit{homomorphism} from $\X$ to $\Y$ is a map $f:X\to Y$ preserving all relations; i.e., $f(\bx)\in R^{\Y}$ for each tuple $\bx\in R^\X$, where $f$ is applied entrywise to the entries of $\bx$. The existence of a homomorphism from $\X$ to $\Y$ is denoted by the expression $\X\to\Y$.
CSPs can be expressed as homomorphism problems for relational structures. In particular, given a fixed $\sigma$-structure $\Y$ (the \textit{template}), the CSP parameterised by $\Y$ is the set of $\sigma$-structures $\X$ (called \textit{instances} of the CSP) for which it holds that $\X\to\Y$.

We will be looking at such homomorphism problems through the lens of the aforementioned two-prover games.
\cref{fig_game_description} illustrates the $\X,\Y$ game associated with an instance $\X$ of $\CSP(\Y)$.
\begin{figure}
\begin{description}
{%
    \item{\sc Input:} $\sigma$-structure $\relstr X$
    \item{\sc Strategy phase:} Alice and Bob strategise
    \item{\sc Proper game phase:}
    \begin{description}
        \item{\sc [Ref]:} sends $R\in\sigma$ and $\tuple x\in R^{\relstr X}$ to Alice
        \item{\sc [Ref]:} sends $x\in X$ to Bob
        \item{\sc [Alice and Bob communicate]}
        (depending on game rules)
        \item{\sc [Alice]:} responds with $\tuple y \in Y^{\ar(R)}$
        \item{\sc [Bob]:} responds with $y\in Y$
    \end{description}
    \item{\sc Scoring:} Alice and Bob win if
    \begin{itemize}
        \item[$(i)$] $\by\in R^\Y$, and
        \item[$(ii)$]  $\forall i\in[\ar(R)]$, $x_i = x$ implies $y_i = y$.
    \end{itemize}
}
\end{description}
\caption{The $\X,\Y$ game.}
\label{fig_game_description}
\end{figure}

Observe that the winning condition $(i)$ corresponds to \textit{plausibility}---Alice's reply should satisfy the local constraint picked by the Referee---while condition $(ii)$ expresses \textit{consistency}---Alice's and Bob's replies should be consistent with each other.

\subsection{Classical strategies---no communication}
Suppose that Alice and Bob are not allowed any communication during the proper game phase, and cannot leverage non-classical effects resulting from shared quantum information to devise stronger strategies.
A \textit{classical strategy} consists then of a function $a_R:R^\X\to Y^{\ar(R)}$ for each $R\in\sigma$ encoding Alice's answer, and a function $b:X\to Y$ encoding Bob's answer. The strategy is \textit{perfect} if it makes Alice and Bob win the game regardless of the Referee's questions; i.e., if $a_R(\bx)\in R^\Y$ for each $\bx\in R^\X$, and $a_R(\bx)_i=b(x_i)$ for each $\bx\in R^\X$ and each $i\in[\ar(R)]$.

We observe that the questions could be thought of as being sampled according to some probability distribution on the instance $\X$. However, since we shall only be concerned with perfect strategies, we do not need to choose any explicit distribution. Randomness could also be involved in Alice's and Bob's replies: In both the classical and the quantum settings (which we shall shortly describe formally), the players could toss a coin before answering. It is not hard to see that, if such a randomised strategy
achieves some winning probability, the same winning probability can be achieved via a fully deterministic strategy. Hence, we shall consider deterministic strategies without loss of generality.

If $\X\to\Y$ (i.e., if $\X$ is a $\YES$-instance of $\CSP(\Y)$), Alice and Bob can agree on some homomorphism $f:\X\to\Y$ in the strategy phase, and then always respond according to $f$: Alice replies $f(\bx)=(f(x_1),\dots,f(x_{\ar(R)}))$ upon receiving $\bx\in R^\X$, and Bob replies $f(x)$ upon receiving $x\in X$. Clearly, this strategy is perfect. Conversely, it is not hard to check that, given a perfect deterministic strategy, Bob's answers yield a homomorphism from $\X$ to $\Y$.
Therefore, we have the following.
\begin{fact}
The set of instances $\X$ for which 
the $\X,\Y$ game admits a perfect classical strategy is $\CSP(\Y)$. 
\end{fact}

\subsection{Quantum strategies---no communication}
\label{subsec_quantum_strategies}
Take now a Hilbert space $\HH$.
We shall always consider finite-dimensional Hilbert spaces, so $\HH\cong\C^d$ or $\HH\cong\R^d$ for some $d\in\N$.
A \textit{projector} onto $\HH$ is a linear map $E:\HH\to\HH$ such that $E^2=E^*=E$, where $E^*$ is the adjoint of $E$. 
Let $I$ and $O$ denote the identity and zero projectors onto $\HH$, respectively.
For a finite set $S$, we denote by $\PVM_{\HH}(S)$ the set of \textit{projection-valued measurements} over $\HH$ whose outcomes are indexed by $S$. In other words, $\PVM_{\HH}(S)$ contains all sets $\mathcal E=\{E_s:s\in S\}$ 
where each $E_s$ is a projector on $\HH$ and, in addition, $\sum_{s\in S}E_s=I$. In particular, the latter condition implies that the projectors in $\mathcal{E}$ are mutually orthogonal: $E_sE_t=O$ when $s\neq t$.

Consider two $\sigma$-structures $\X$ and $\Y$. A \textit{quantum strategy} for the $\X,\Y$ game consists of a Hilbert space $\HH$, a unit vector $\psi\in\HH\otimes\HH$ (which describes the state of the quantum system shared by Alice and Bob), a measurement $\mathcal{A}^{(\bx)}\in\PVM_\HH(Y^{\ar(R)})$ for each $R\in\sigma$ and each $\bx\in R^\X$, and a measurement $\mathcal{B}^{(x)}\in\PVM_\HH(Y)$ for each $x\in X$.
Upon receiving $\bx\in R^\X$, Alice uses $\mathcal{A}^{(\bx)}$ to measure her part of the shared quantum system; upon receiving $x\in X$, Bob uses $\mathcal{B}^{(x)}$ to measure his part of the system. This results in Alice replying $\by$ and Bob replying $y$ with probability $\psi^*(A^{(\bx)}_\by\otimes B^{(x)}_y)\psi$. Like in the classical case, the strategy is winning if the resulting answers always satisfy the conditions $(i)$ and $(ii)$ of~\cref{fig_game_description}.

Let $\CSP^\HH(\Y)$ be the set of $\sigma$-structures $\X$ such that a perfect quantum strategy for the $\X,\Y$ game exists. 
It was shown in~\cite{abramsky2017quantum} that $\CSP^\HH(\Y)$ is in fact an (infinite-domain) classical $\CSP$, parameterised by a $\sigma$-structure $\Y^\HH$ that we now describe.
The domain of $\Y^\HH$ is $\PVM_{\HH}(Y)$;
for each symbol $R\in\sigma$, a tuple $(\mathcal{E}^{(1)},\dots,\mathcal{E}^{(\ar(R))})$ of PVMs belongs to $R^{\Y^\HH}$ if and only if there exists some $\mathcal{F}=\{F_{\by}\}_{\by\in R^\Y}\in\PVM_{\HH}(R^\Y)$ such that the condition 
\begin{align}
\label{eqn_defining_quantum_structure}
E^{(i)}_{y}
=
\sum_{\substack{\by\in R^\Y\\y_i=y}}F_{\by}
\end{align}
holds for each $y\in Y$ and each $i\in[\ar(R)]$.

\begin{thm}[\cite{abramsky2017quantum}]
\label{thm_quantum_csp_is_csp}
$\CSP^\HH(\Y)=\CSP(\Y^\HH)$ for each $\sigma$-structure $\Y$ and each Hilbert space $\HH$. 
\end{thm}

\subsection{Classical strategies---Alice messages Bob}
\label{sec_description_channel_strategies_alice}
Consider again the game of~\cref{fig_game_description}. This time, suppose that, after receiving her question $\bx\in R^\X$ from the Referee, Alice is allowed to send a message to Bob, with the goal to provide him with as much information about $\bx$ as possible so that their answers are winning.
Clearly, if the message is allowed to encode the complete information about $\bx$, Bob can always respond perfectly: He knows that Alice will reply with some $\by\in R^\X$ upon receiving $\bx$, as they agreed in the strategy phase. Then, he can simply reply with $y_i$ whenever $x=x_i$ for some index $i$, and with any answer otherwise.
Hence, provided that Alice has some plausible reply at all,\footnote{For example, if Alice is sent a loop of $\X$ but $\Y$ has no loop, she cannot give a plausible answer. Via~\cref{defn_alice_power}, it shall be immediate to check when this happens.} the players always win in this case.

The situation becomes more interesting if the message is required to have a \textit{limited} length, independent of $\X$.
Suppose, for concreteness, that Alice is only allowed to transfer $\log_2 k$ classical bits. In other words, she can send Bob one of $k$ different prescribed messages.
In this setting, a (classical) strategy consists of a function $a_R:R^\X\to Y^{\ar(R)}$ and a function $m_R:R^\X\to[k]$ for each $R\in\sigma$, as well as a function $b:X\times[k]\to Y$. Here, $a_R$ encodes Alice's reply to the Referee; $m_R$ encodes her message to Bob; and $b$ encodes Bob's reply (which depends on both the Referee's question and Alice's message). 
The conditions making such a strategy \textit{perfect} are the obvious ones.

Let $\GameAB{k,\Y}$ denote the set of instances $\X$ for which a perfect classical strategy with Alice sending Bob one of $k$ prescribed messages (i.e., $\log_2 k$ bits) exists.
Quite surprisingly, we shall see in~\cref{sec_combinatorial_characterisation_strategies} that,
just like in the quantum case, $\GameAB{k,\Y}$ is a CSP, parameterised by a certain combinatorial power of $\Y$ that we denote by $\uprelstr{k,\Y}$. Using this characterisation, we obtain in~\cref{sec_entangled_vs_classical} our first main result, stating that any perfect quantum strategy can be simulated by a classical strategy with a suitably long ``Alice to Bob'' message:
\begin{restatable}{thm}{mainAlice}
\label{thm_main_alice_general}
    For each $\sigma$-structure $\Y$ and each Hilbert space $\HH$ there exists some $k\in\N$ such that $\CSP^\HH(\Y)\subseteq\GameAB{k, \Y}$.
\end{restatable}

Since, clearly, $\CSP(\Y)\subseteq\CSP^\HH(\Y)$,~\cref{thm_main_alice_general} establishes a ``sandwich'' of $\CSP^\HH(\Y)$ between two finite-domain CSPs corresponding to classical strategies.
Our proof of~\cref{thm_main_alice_general} does not yield an explicit bound on $k$ as a function of $\HH$ and $\Y$. On the other hand, in the specific setting of digraphs (i.e., relational structures having a single, binary relation), we show that, asymptotically, $k=\dim\HH+\log\chi(\Y)$ is enough, where $\chi(\Y)$ is the chromatic number of $\Y$ (see~\cref{thm_main_digraphs}). Hence, for a fixed $\chi(\Y)$, the number of communication bits in the classical strategy matches the number of shared qubits in the quantum strategy---we find this fact quite fascinating. Our proof for this case makes use of a geometric result by Rogers on the number of \textit{caps} needed to cover a real sphere in arbitrary dimension~\cite{rogers1963covering}.

\subsection{Classical strategies---Bob messages Alice}
What if the message is sent by Bob to Alice instead?
In this case, a strategy consists of a function $b:X\to Y$ encoding Bob's reply to the Referee, a function $m:X\to[k]$ encoding Bob's message to Alice, and, for each $R\in\sigma$, a function $a_R:R^\X\times[k]\to Y^{\ar(R)}$ encoding Alice's answer to the Referee after reading Bob's message.

Although this scenario might look similar to the previous one at first glance, it turns out that there exist substantial differences.
In sharp contrast to the previous case, it is easy to find CSP instances for which no amount of ``Bob to Alice'' communication can make them win the game with certainty, even if plausible answers are available for every question.
\begin{example}
\label{example_K2_vs_D2}
Let $\X$ be an undirected edge  $x\leftrightarrow x'$ and $\Y$ be a directed edge $y\to y'$.
Consider the $\X,\Y$ game.
Suppose Bob receives $x$ from the Referee. If he answers $y$, Alice will never be able to give a winning reply upon receiving the edge $x'\to x$ from the Referee, no matter what message she receives. Similarly, if Bob answers $y'$, Alice will not be able to find a convincing answer to the question $x\to x'$.
\end{example}

We will denote by $\GameBA{k, Y}$ the set of instances $\relstr{X}$ for which a perfect classical strategy with Bob sending Alice one of $k$ prescribed messages exists. In this case, we obtain a similar characterisation of $\GameBA{k, Y}$ as the CSP parameterised by a structure that we denote by $\uprelstrbob{k,\Y}$. Through this characterisation, we show that this dual messaging protocol is also able to capture quantum strategies, provided that the length of the message is sufficiently long:

\begin{restatable}{thm}{mainBob}
\label{thm_main_bob_general}
    For each $\sigma$-structure $\Y$ and each Hilbert space $\HH$ there exists some $k\in\N$ such that $\CSP^\HH(\Y)\subseteq\GameBA{k, \Y}$.
\end{restatable}

The ``Bob messages Alice'' protocol appears to be weaker than the ``Alice messages Bob'' one. For example, the explicit upper bound on $k$ that we obtain in the digraph case for the former protocol (see~\cref{thm_main_digraphs}) is only exponential in $\dim(\HH)$ (however, see Remark~\ref{rem_bob_sometimes_stronger}).

\section{Capturing classical strategies with extra communication}
\label{sec_combinatorial_characterisation_strategies}

In this section, we show that both ``Alice to Bob'' and ``Bob to Alice'' perfect classical strategies can be captured by suitable combinatorial  powers of the template $\Y$. We handle the two cases separately.

\subsection{Alice messages Bob}
\label{subsec_alice_messages_bob}
We start by defining the combinatorial operation capturing ``Alice to Bob'' strategies.

\begin{defn}
\label{defn_alice_power}
    Let $\relstr Y$ be a $\sigma$-structure and $k$ a positive integer.
    The $k$-fold \textit{Alice power} of $\relstr Y$---in symbols, $\uprelstr{k,\Y}$---is the $\sigma$-structure on the universe $Y^k$ such that, for every relation symbol $R$~(say of arity $r$), a tuple $(\tuple x^{(1)}, \dotsc, \tuple x^{(r)})$ belongs to $R^{\uprelstr{k,\Y}}$ if and only if there exists $j\in[k]$ such that $(x_j^{(1)},\dotsc,x_j^{(r)})\in R^{\relstr Y}$. 
\end{defn}

In the case of graphs, $\uprelstr{k,\Y}$ corresponds to the $k$-fold iteration of the \textit{Cartesian sum} of $\Y$ with itself---an operation introduced by Ore~\cite{ore1962theory} and also known in the literature under the name of \textit{co-normal product}~\cite{frelih2013edge}.
We shall prove the following.

\begin{thm}
\label{thm_algebraic_characterisation_alice_to_bob}
$\GameAB{k,Y} = \MKCSP{\uprelstr {k, \Y}}$ for each $\sigma$-structure $\Y$ and each $k\in\N$. 
\end{thm}
\begin{proof}
   We begin by proving the right-to-left inclusion for fixed (but arbitrary) $k$ and $\relstr Y$.
   Suppose that $\relstr X$ maps to $\uprelstr {k,\Y}$ via a homomorphism $h$. Alice and Bob fix one such homomorphism in the strategising phase.
   
   In the game phase Alice, given a tuple $\tuple x\in R^{\relstr X}$, computes  $h(\tuple x)$. 
   By the choice of $h$, it must be the case that $h(\tuple x)$ belongs to $ R^{\uprelstr{k,\Y}}$, and, by the definition of $\uprelstr {k, \Y}$, there exists $j$ so that $j$-th projection of $h(\tuple x)$ belongs to $R^{\relstr Y}$.
   Alice sends this projection to the Referee and sends $j$ to Bob.
   Bob receives $x$ from the Referee, receives $j$ from Alice, and sends $h(x)_j$ to the Referee.
   Their responses are consistent, so they win the game.

   For the other inclusion:
   Fix $\relstr X$ such that Alice and Bob win the game on $\relstr X$.
   Let us look at Bob's strategy: Bob is given $x\in X$ from the Referee and $j$ from Alice, and he plays $y\in Y$. 
   We use these values to define a map $h:X\rightarrow Y^k$ by simply setting $h(x)_j = y$~%
   (by the definition of the game, each $x$ will be given to Bob, %
   and if Alice never sends $j$ we set $h(x)_j$ to be any $y\in Y$). 
   It remains to verify that $h$ is a homomorphism from $\relstr X$ to $\uprelstr {k, \Y}$. 
   By the definition of a homomorphism, we fix $R^{\relstr X}$ and a tuple $\tuple x\in R^{\relstr X} $.
   These might have been chosen as the Referee's question to Alice, so there is a single $\tuple y$---Alice's answer to $\tuple x\in R^{\relstr X}$.
   Moreover, there is a fixed value $j$ which Alice sends to Bob in this case.
   Since the strategy is perfect, given $x_i$ and this $j$, Bob needs to respond with $y_i$.
   However, this implies that the projection of $h(\tuple x)$ on the $j$-th coordinate is equal to $\tuple y$, and since $\tuple y\in R^{\relstr Y}$ the map is a homomorphism.
   The proof is concluded.
\end{proof}

\begin{example}
\label{example_AliceToBob_cliques}
We illustrate \cref{thm_algebraic_characterisation_alice_to_bob} with two examples, whose correctness the reader can easily verify.
\begin{itemize}
    \item Let $\K_n$ be the clique on $n$ vertices. Then
$\uprelstr{ k, \K_n} =
\relstr K_{n^k}$.
\item Let $\NAE_{n}$ (read ``Not All Equal'') be the structure on domain $[n]$ having a single, ternary relation containing all non-constant tuples. Then 
$\uprelstr{k, \NAE_{2}} = \NAE_{2^k}$.
\end{itemize}
\end{example}

Clearly, the ability to transfer information will never disadvantage Alice and Bob. 
\cref{thm_algebraic_characterisation_alice_to_bob} makes this immediate: The map $y\mapsto (y,y)$ is a homomorphism from $\relstr Y$ to $\uprelstr {2,\Y}$, so 
$\GameAB{1,Y}=\MKCSP{\relstr Y}\subseteq\MKCSP{\uprelstr {2,\Y}}=\GameAB{2, Y}$.
More generally,
\begin{equation*}
    \MKCSP{\relstr Y} = \GameAB{1,\relstr Y}
    \subseteq \GameAB{2, \relstr Y}
    \subseteq \GameAB{s, \relstr Y} \subseteq \dotsb
\end{equation*}

On the other hand, does---say---one bit of communication help?
The next theorem states that, except for degenerate cases, it does.

Recall that the \textit{core} of a $\sigma$-structure $\mathbf{Y}$---in symbols, $\core{\Y}$---is a minimal (under inclusion) induced substructure of $\Y$ that is homomorphically equivalent to $\Y$. It is easy to check that $\core{\Y}$ is unique up to isomorphism, and any endomorphism is an automorphism~\cite{BKW17}.

\begin{thm}\label{thm:everybithelps}
    \sloppy Given a $\sigma$-structure $\Y$, we have
    $\GameAB{2, \relstr Y}=\CSP(\Y)$
    if and only if the domain of $\core{\relstr Y}$ has size $\leq 1$.
\end{thm}
\noindent Before proving this, we establish a simple lemma.
\begin{lem}
\label{lem_basic_alice}
    For every $k\in\N$, if $\relstr X \rightarrow \relstr Y$ then $\uprelstr{k,\relstr X}\rightarrow \uprelstr{k,\relstr Y}$.
    In particular $\uprelstr {k, \relstr Y}$ and $\uprelstr{k, \core{\relstr Y}}$
    are homomorphically equivalent (and, thus define the same CSP).
\end{lem}
\begin{proof}
    For the first claim, it suffices to build a homomorphism coordinatewise.
    The second claim is then immediate.
\end{proof}
\begin{proof}[Proof of \cref{thm:everybithelps}]
    Note that the condition $\GameAB{2, \relstr Y}=\CSP(\Y)$ is equivalent to $\uprelstr {2,\Y}\to\relstr Y$, since $\Y\to\uprelstr {2,\Y}$ always holds. Moreover,
    by~\cref{lem_basic_alice}, we can substitute $\relstr Y$ with $\core{\relstr Y}$, and the games do not change.
    
    The implication from right to left is obvious now: If $|Y|\leq 1$ then $\uprelstr{2, \relstr Y}$ is isomorphic to $\relstr Y$ and we gain nothing by communication.

    For the converse implication, we fix a core $\relstr Y$ and assume that $\uprelstr{2, \relstr Y}\to\relstr Y$. 
    The first step is to notice that every permutation of $Y$ is an isomorphism of $\relstr Y$:
    Indeed, fix such a permutation $\sigma$ and consider elements of $\uprelstr{ 2,\relstr Y}$ of the form $(y,\sigma^{-1}(y))$.
    Let $\relstr Y'$ be substructure of $\uprelstr {2, \relstr Y}$ induced on these elements.
    Note that $\relstr Y\rightarrow \relstr Y'$~(by the first coordinate) and that $\relstr Y'\rightarrow \relstr Y$~%
    (since even the larger structure $\uprelstr{2,\relstr Y}$ homomorphically maps to $\Y$, by assumption).
    Since $\relstr Y$ and $\relstr Y'$ have the same number of elements,
    and since $\relstr Y$ is a core, the structures need to be isomorphic.
    This, in turn, implies that the number of tuples in the relations of $\relstr Y$ and $\relstr Y'$ needs to be the same.
    In particular, if $\sigma^{-1}(\tuple y)\in R^{\relstr Y}$ then $\tuple y\in R^{\relstr Y}$.
    Applying $\sigma$ twice we find that $\tuple y\in R^{\relstr Y}$ implies $\sigma(\tuple y)\in R^{\relstr Y}$; i.e., $\sigma$ is an isomorphism of $\Y$.

    As a consequence of the discussion above, if a relation $R^{\relstr Y}$ contains a constant tuple, then it contains all constant tuples.
    \sloppy Furthermore, if every relation contains every constant tuple, then $\core{\relstr Y}$~%
    (which is equal to $\relstr Y$) is a single element, and the conclusion follows.
    It remains to derive a contradiction in the remaining case---when $R^{\relstr Y}$ has
    no constant tuples for some symbol $R$.

    Let $\tuple y\in R^{\relstr Y}$ have a minimal number of distinct elements among the tuples of $R^{\relstr Y}$.
    Find two different elements of $\uprelstr {2,Y}$:
    $(b,c)$ and $(d,e)$ so that $h(b,c)=h(d,e)$.
    Note that these must exist just by counting the domain sizes.
    Assume, without loss of generality, that they differ on the first coordinate.
    We choose two different elements that appear (perhaps multiple times) in $\tuple y$ and substitute the appearances of the first one with $b$ and of the second one with $d$.
    The new tuple is still in $R^{\relstr Y}$ by the symmetry from the previous paragraph.
    Next we extend the new tuple to a tuple $\tuple y'\in R^{\uprelstr{2, \relstr Y}}$.
    We conclude that $h(\tuple y')\in R^{\relstr Y}$ has fewer distinct elements than $\tuple y$---a contradiction.
    This finishes the proof.
\end{proof}

\subsection{Bob messages Alice}
As noted in~\cref{sec_strategies}, switching the direction of the communication between Alice and Bob substantially changes the power of the corresponding strategies.
This will be evident by looking at the structure capturing strategies of the second type, described below, which turns out to be structurally different from the one of~\cref{subsec_alice_messages_bob}. For example, this time, the domain size of the template $\Y$ grows linearly rather than exponentially in $k$. As we shall see in~\cref{sec_entangled_vs_classical}, this fact will force the number of communication bits necessary to simulate a quantum strategy to be much larger in the ``Bob to Alice'' case.

\begin{defn}
\label{defn_bob_power}
    Let $\relstr Y$ be a $\sigma$-structure and $k$ a positive integer.
    The $k$-fold \textit{Bob power} of $\Y$---in symbols, $\uprelstrbob{k,\Y}$---is the $\sigma$-structure on the universe $[k] \times Y$ such that, for every relation symbol $R$ (say of arity $r$), a tuple $((i_1, y_1), \dotsc, (i_r, y_r))$ belongs to $R^{\uprelstrbob{k, \Y}}$ if and only if for each $s \in [k]$ there exists some tuple $\tilde \by=(\tilde y_1,\dotsc,\tilde y_r)\in R^{\relstr Y}$ such that the implication $i_j = s \implies y_j = \tilde y_j$ holds for each $j\in[r]$.
\end{defn}
This operation is a generalisation of the standard \textit{join} operation for undirected graphs. In particular, when $\Y$ is an undirected graph with no isolated vertices, $\uprelstrbob{2,\Y}$ is exactly the join of two copies of $\Y$.
Like in the case of $\GameAB{k,Y}$, we now show that $\GameBA{k,Y}$ is a CSP for a certain template, corresponding to the structure defined above. The proof of the next result is deferred to~\cref{app_omitted_proofs}.

\begin{restatable}{thm}{thmalgebraiccharacterisationbobtoalice}
\label{thm_algebraic_characterisation_bob_to_alice}
$\GameBA{k,Y} = \CSP({\uprelstrbob{k, \Y}})$ for each $\sigma$-structure $\Y$ and each $k\in\N$.
\end{restatable}

\begin{example}
    Recall the structures $\K_n$ and $\NAE_n$ from~\cref{example_AliceToBob_cliques}. The following straightforward facts hold.
\begin{itemize}
    \item $\uprelstrbob{ k, \K_n} = \relstr K_{nk}$.
    \item $\uprelstrbob{k, \NAE_2} = \NAE_{2k}$.
    \item Let $\mathbf{RB}_{n}$ (read ``Rainbow'') be the structure on domain $[n]$ having a single, ternary relation containing all tuples that contain three distinct vertices. Then 
$\uprelstrbob{k,\mathbf{RB}_3} = \mathbf{RB}_{3k}$.
\end{itemize}
\end{example}

Also in this communication setting, we have the following chain of inclusions:
\begin{equation*}
    \MKCSP{\relstr Y} = \GameBA{1,\relstr Y}
    \subseteq \GameBA{2, \relstr Y}
    \subseteq \GameBA{s, \relstr Y} \subseteq \dotsb
\end{equation*}

However, does---say---one bit of communication help the provers?
Unlike in the ``Alice to Bob'' case, not always.
Consider the digraph $\mathbf{D}_2$ consisting of a single directed edge. 
Then $\uprelstrbob{k,\mathbf{D}_2}$ is homomorphically equivalent to $\mathbf{D}_2$ for any $k$.

In fact, we can characterise the structures $\Y$ for which one bit of communication gives Alice and Bob an advantage.

\begin{restatable}[Proved in \cref{app_omitted_proofs}]{thm}{thmwhenbithelpsbob}
\label{thm:whenbithelpsbob}
    \sloppy Given a $\sigma$-structure $\Y$, we have $\GameBA{2, \relstr Y}=\CSP(\Y)$
    if and only if every relation in $\core\Y$ is a Cartesian product of some subsets of $Y$.
\end{restatable}

    Observe that the condition in~\cref{thm:whenbithelpsbob} is equivalent to the fact that $\CSP(\Y)=\CSP(\Y')$ for some $\Y'$ having only unary relations~\cite{BKW17}. Such CSPs are trivially tractable, e.g. using the 1-minimality algorithm from \cite{BKW17}.
As a consequence, for any CSP that is not tractable in polynomial time, classical strategies with extra communication provide an actual advantage with respect to classical strategies with no communication. The same fact was observed for quantum strategies (over a Hilbert space of dimension at least $3$) in~\cite[Cor.17]{ciardo_quantum_minion}.  

\begin{cor}
If $\CSP(\Y)$ is not tractable in polynomial time, then any additional communication gives Alice and Bob advantage. In other words, in this case, $\CSP(\Y) \varsubsetneq \GameAB{2,\Y}$ and $\CSP(\Y) \varsubsetneq\GameBA{2,\Y}$.
\end{cor}

%
\begin{rem}
\label{rem_bob_sometimes_stronger}
Alice messaging Bob usually appears to be a better option for winning the game than Bob messaging Alice. For example, as noted in~\cref{example_K2_vs_D2}, if $\mathbf{D}_2$ is a single directed edge and $\K_2$ is a single undirected edge, no amount of ``Bob to Alice'' communication will allow them to win, while even one bit of ``Alice to Bob'' communication is enough. Indeed, it is easy to see that $\K_2\to\GameAB{2, \relstr \mathbf{D}_2}$ while $\K_2\not\to\GameBA{k, \relstr \mathbf{D}_2}$ for any $k\in\N$.
 
On the other hand, surprisingly, there are cases where Bob messaging Alice is a better choice. Consider two Boolean structures $\X$ and $\Y$ with a single $4$-ary relation $R^\X=\{(0, 0, 1, 1)\}$ and $R^\Y=\{(0,1,1,1),(1,1,1,0)\}$.
It is not hard to verify that $\X\to\uprelstrbob {2,\Y}$.
%
%
%
In terms of strategy, suppose Bob is allowed to send a single bit (i.e., $0$ or $1$) to Alice. Bob always answers $1$ to the Referee and sends his input to Alice; Alice responds with $(1, 1, 1, 0)$ upon receiving $0$ from Bob, and $(0, 1, 1, 1)$ otherwise, thus winning the game. However, it is straightforward to check that $\X\not\to\uprelstr{k, \Y}$ for any $k\in\N$.
\end{rem}


\section{Quantum vs. classical strategies}
\label{sec_entangled_vs_classical}

\sloppy In this section, we prove our main results---\cref{thm_main_alice_general} and~\cref{thm_main_bob_general}---which state that any perfect quantum strategy can be simulated by a perfect classical strategy with an ``Alice to Bob'' or ``Bob to Alice'' message whose size depends only on $(i)$ the dimension of the shared quantum system appearing in the quantum strategy, and $(ii)$ structural parameters of the template structure $\Y$. We stress that it is the lack of dependence of the message size on the CSP instance to make our results non-trivial: Any quantum measurement scenario in which a player performs one of $m$ different measurements can be easily perfectly simulated by a classical protocol augmented with $\log_2 m$ communication bits, as noted in~\cite[Thm.~1]{brassard1999cost}. 

To prove our main results, we make use of the characterisations given in \cref{sec_combinatorial_characterisation_strategies}.
It turns out that the ``Alice to Bob'' protocol is substantially more effective at simulating quantum strategies, as the required message length is exponentially smaller in this case. 
While the results hold for arbitrary CSPs, we first consider the special case of digraphs---i.e., structures having
a unique relation of arity two. This is for two reasons: First, the treatment in the digraph case is simpler and still provides intuition on the general case; second, for digraphs, we are able to achieve an explicit and rather tight upper bound on the amount of classical communication  needed to simulate quantum strategies. In particular, in the ``Alice to Bob'' case, for a given template $\Y$ the number of required communication bits matches asymptotically the number of shared qubits in the quantum strategy that is being simulated (see part $(i)$ of~\cref{thm_main_digraphs}).

\subsection{Binary case}

We establish the following result, which is a refinement of ~\cref{thm_main_alice_general} and~\cref{thm_main_bob_general} for the digraph case.
We denote by $\chi(\Y)$ the (classical) chromatic number of $\Y$---i.e., the minimum $n\geq 2$ such that $\Y\to\K_n$.

\begin{thm}
\label{thm_main_digraphs}
For each digraph $\Y$ and each Hilbert space $\HH$ of dimension $d$ it holds that
\begin{itemize}
    \item[$(i)$] $\CSP^\HH(\Y)\subseteq\GameAB{k, \Y}$ with $k=d+\log\chi(\Y)+\mathcal{o}(d)+\mathcal{o}(\log\chi(\Y))$;
    \item[$(ii)$] $\CSP^\HH(\Y)\subseteq\GameBA{k, \Y}$ with $k=\chi(\Y) \cdot \mathcal O(\exp(d))$.
\end{itemize}
\end{thm}

The proof strategy can be summarised via the chain of homomorphisms
\begin{equation*}
\Y^\HH\to\K_n^\HH\to\K_m\to \uprelstr{k,\Y}
\end{equation*}
for suitable choices of $n$ and $m$.
The first link in the chain is covered by the fact that the map $\Y\mapsto\Y^\HH$ is functorial for arbitrary relational structures---in fact, as established in~\cite{abramsky2017quantum}, the map is even monadic.
\begin{prop}[\cite{abramsky2017quantum}]
\label{YH_monadic}
    Let $\Y,\Y'$ be $\sigma$-structures, and let $\HH$ be a Hilbert space. Then $\Y\to\Y'$ implies $\Y^\HH\to\Y'^\HH$.
\end{prop}
Observe that the above result also holds for Alice and Bob powers (i.e., for ``Alice to Bob'' and ``Bob to Alice'' classical strategies), as shown in~\cref{lem_basic_alice} and~\cref{lem_basic_bob}, respectively.

For the second link in the chain of homomorphisms, we will need the following proposition, which is the most technical part of the proof.\footnote{In the statement below, as well as in the statement of \cref{thm_rogers_spheres}, $\mathcal{o}(1)$ indicates a function of $d$ that approaches $0$ as $d\to\infty$.}

\begin{prop}
\label{prop_QKn_to_Km}
    Let $\HH$ be a Hilbert space of dimension $d$ and let $n\geq 2$ be an integer. Then $\K_n^\HH\to\K_m$
    with $m=n(2+\mathcal{o}(1))^{d-1}$.
\end{prop}

The proof of \cref{prop_QKn_to_Km} relies on some additional facts.
A classic geometric result from~\cite{rogers1963covering} provides a technique for covering spheres in $\R^d$ with few \textit{sphere caps}.
\sloppy It was noted in~\cite{prosanov2018chromatic} that Rogers' method allows colouring spheres with few colours in such a way that monochromatic points are at a distance other than 1.
For each $d\in\N$ and each $0<\rho\in\R$, let $S(\rho,d)$ denote the $(d-1)$-dimensional sphere in $\R^d$ of radius $\rho$.

\begin{thm}[\cite{rogers1963covering,prosanov2018chromatic}]
\label{thm_rogers_spheres}
    It is possible to colour $S(\rho,d)$ with $(2\rho+\mathcal{o}(1))^{d-1}$ colours in such a way that points at distance $1$ have distinct colours.
\end{thm}

A second result that shall be useful comes from the analysis of the \textit{quantum minion} from~\cite{ciardo_quantum_minion}. We present it here in the simplified setting of cliques, as we do not need the higher level of generality.
We say that a complex- or real-valued matrix $M$ is a \textit{frame} if $MM^*$ is a diagonal matrix of trace $1$.
Consider the infinite graph $\Sminion_{\C}(d,n)$ whose domain is the set of $n\times d$ complex frames, with two frames $M,M'$ forming an edge if and only if there exists an $(n^2-n)\times d$ complex frame $N$ such that
\begin{align}
\label{eqn_1133_1901}
M_i=\sum_{j\neq i}N_{(i,j)},\quad  M'_i=\sum_{j\neq i}N_{(j,i)} 
\end{align}
for each $i\in [n]$, where $M_i$ denotes the $i$-th row of $M$ and $N_{(i,j)}$ denotes the $(i,j)$-th row of $N$.
Let also $\Sminion_{\R}(d,n)$ denote the subgraph of $\Sminion_{\C}(d,n)$ induced by real-valued frames. 
The next result is a consequence of~\cite[Prop.22--Thm.24]{ciardo_quantum_minion}.

\begin{lem}[\cite{ciardo_quantum_minion}]
\label{lem_QH_vs_Sminion}
For any Hilbert space $\HH$ of dimension $d$ and any integer $n\geq 2$, it holds that
\begin{align*}
    \K_n^\HH\to\Sminion_{\C}(d,n)\to\Sminion_{\R}(2d,n).
\end{align*}
\end{lem}

We are now in a position to prove \cref{prop_QKn_to_Km}.

\begin{proof}[Proof of \cref{prop_QKn_to_Km}]
Recall that $d$ denotes the dimension of the Hilbert space $\HH$, and
consider the sphere $S(\rho,2d)$ with $\rho=\frac{1}{\sqrt{2}}$. By \cref{thm_rogers_spheres}, there exists a colouring $\alpha$ of this sphere using $(\sqrt{2}+\mathcal{o}(1))^{2d-2}$ colours such that points at distance $1$ get distinct colours. By the choice of $\rho$, this means that $\alpha$ is an \textit{orthogonal colouring}---i.e., no two points of the sphere are coloured with the same colour when the corresponding vectors are orthogonal.
Take a frame $M\in\Sminion_{\R}(2d,n)$, and find some index $i\in[n]$ such that the $i$-th row of $M_i$ is nonzero. Observe that such a row exists as $\tr(MM^*)=1$. 
Let also $\bv=\frac{M_i}{\sqrt{2}\norm{M_i}}$, and observe that $\bv\in S(\rho,2d)$.
We assign to $M$ the pair $(i,\alpha(\bv))$. We claim that this defines a proper $m$-colouring of $\Sminion_{\R}(2d,n)$---where we implicitly identify $[m]$ with $[n]\times[(\sqrt{2}+\mathcal{o}(1))^{2d-2}]$. Suppose that two frames $M,M'\in\Sminion_{\R}(2d,n)$ are coloured with the same colour $(i,a)$. Let $\bv$ and $\bv'$ be the normalised $i$-th rows of $M$ and $M'$, respectively. We have that $\alpha(\bv)=\alpha(\bv')$.
Since $\alpha$ is a colouring, it follows that $\bv\not\perp \bv'$; i.e., $M_i\not\perp M_i'$. Suppose, for the sake of contradiction, that $M$ and $M'$ are adjacent in $\Sminion_{\R}(2d,n)$. This means that there exists an $(n^2-n)\times 2d$ frame $N$ satisfying~\eqref{eqn_1133_1901}. Since $NN^*$ is diagonal, it follows that there exists some row $N_{(k,\ell)}$ of $N$ appearing in both sums. But this implies that $k=i$, $\ell=i$, and $k\neq\ell$, a contradiction.
Hence, the map constructed above is a proper $m$-colouring of $\Sminion_{\R}(2d,n)$; i.e., $\Sminion_{\R}(2d,n)\to\K_m$. Composing this colouring with the homomorphisms of~\cref{lem_QH_vs_Sminion} yields the required colouring $\K_n^\HH\to\K_m$.
\end{proof}

We can now show that perfect quantum strategies for digraphs can be turned into perfect classical strategies with a small communication channel. 

\begin{proof}[Proof of~\cref{thm_main_digraphs}]
First of all, we can assume that $\Y$ is loopless, as otherwise both digraphs $\uprelstr{k,\Y}$ and $\uprelstrbob{k,\Y}$ contain a loop and, hence, $(i)$ and $(ii)$ trivially hold.
Furthermore, if $\Y$ homomorphically maps to a directed edge, $\CSP(\Y)$ has bounded width. In this case, it was shown in~\cite[Cor.~25]{ciardo_quantum_minion} that $\X\to\Y^\HH$ is equivalent to $\X\to\Y$, which means that $\CSP^\HH(\Y)=\CSP(\Y)$. Thus, the theorem also holds trivially in this case.
Assume from now on that $\Y \not\to \mathbf D_2$.

Write $n=\chi(\Y)$.
Using~\cref{YH_monadic} and~\cref{prop_QKn_to_Km},  we obtain
\begin{align}
    \label{eqn_1247_1901}\Y^\HH\to\K_n^\HH\to\K_m,
\end{align}
with
$m=n(2+\mathcal{o}(1))^{d-1}$.

Take some $k\in\N$.
Pick a directed edge $(y_1, y_2) \in R^\Y$ and the subset $S \subseteq \{y_1, y_2\}^k$ consisting of all tuples in which the number of occurrences of $y_1$ is exactly $\lfloor k / 2\rfloor$. Thus, $|S| = {k \choose \lfloor k/2\rfloor}$, and for any two distinct tuples $\textbf{t}^{(1)}, \textbf{t}^{(2)} \in S$ there exist indices $i, j \in [k]$ such that $t^{(1)}_i = y_1, t^{(2)}_i = y_2$ and $t^{(1)}_j = y_2, t^{(2)}_j = y_1$. It follows that both $(\textbf{t}^{(1)}, \textbf{t}^{(2)})$ and $(\textbf{t}^{(2)}, \textbf{t}^{(1)})$ are in $R^{\uprelstr{k, \Y}}$.

By relaxing Stirling's approximation of the central binomial coefficient, we have $|S| = {k \choose \lfloor k/2\rfloor} \geq \frac{2^{k - 1}}{\sqrt{k}}$, which matches exactly for $k = 1$ and diverges for $k\geq 2$. Now, by picking $k = \lceil \log_2 m + \log_2\log_2 m + 2\rceil$, we get
\begin{equation*}
    |S| \geq \frac{2 m \log_2 m}{\sqrt{\log_2 m + \log_2\log_2 m + 3}} \geq m
\end{equation*}
(since the function $f(x) = 4 \log_2^2 x - \log_2 x - \log_2 \log_2 x - 3$ is increasing for $x \geq 2$ and $f(2) = 0$).
We can then complete the chain of homomorphisms in~\cref{eqn_1247_1901} with
\begin{align*}
    \K_m\to\uprelstr{\log_2 m + \mathcal{O}(\log\log m),\Y},
\end{align*}
which implies part $(i)$ of the statement.

As for part $(ii)$, recall that we work now under the assumption that $\Y \not\to \mathbf D_2$. This implies that there exists some vertex in $\Y$ having both in-degree and out-degree greater than zero. This means that $\uprelstrbob{m,\Y}$ contains a clique of size $m$, so~\eqref{eqn_1247_1901} can be completed with
$\K_m\to\uprelstrbob{m,\Y}$,
and $(ii)$ follows.
\end{proof}

\subsection{General case}

We now generalise \cref{thm_main_digraphs} to all CSPs, thus proving the main results of this paper, which we restate here for convenience.

\mainAlice*
\mainBob*

First, we introduce some terminology.
Given an integer $r\in\N$, we let an \textit{$r$-pattern} be a set of partitions of the set $[r]$. For a signature $\sigma$, a \textit{$\sigma$-pattern} is a collection $\mathcal{P}=(P_R)_{R\in\sigma}$, where each $P_R$ is an $\ar(R)$-pattern. 
Given two partitions $\pi,\pi'$ of $[r]$, we write $\pi\preceq\pi'$ (read: ``$\pi$ is at least as fine as $\pi'$'') if each part in $\pi$ is a subset of a part in $\pi'$. 
Given a set $S$ and a tuple $\bs\in S^r$, we let $\pi_\bs$ be the partition of $[r]$ consisting of the set of equivalence classes of $[r]$ modulo the equivalence relation defined by $i\sim j$ if and only if $s_i=s_j$.
For a $\sigma$-structure $\Y$, we let $\ptn(\Y)$ be the $\sigma$-pattern $(P_R)_{R\in\sigma}$ where, for each $R\in\sigma$, $P_R$ is the set of partitions $\pi$ of $[\ar(R)]$ such that $\pi\preceq\pi_\by$ for some $\by\in R^\Y$.

Take a $\sigma$-pattern $\mathcal{P}=(P_R)_{R\in\sigma}$ and an integer $n\in\N$.
We define the \textit{complete structure} of size $n$ on pattern $\mathcal{P}$ as the $\sigma$-structure $\K_{n,\mathcal{P}}$ having domain $[n]$ and relations $R^{\K_{n,\mathcal{P}}}=\{\ba\in [n]^{\ar(R)}\mid\pi_\ba\preceq \pi\mbox{ for some }\pi\in P_{R}\}$, for each $R\in\sigma$. We let the \textit{chromatic number} of a $\sigma$-structure $\Y$ be 
\begin{align*}
    \chi(\Y)=\min\{n\in\N\mid \Y\to\K_{n,\ptn(\Y)}\}.
\end{align*}
Observe that $\chi(\Y)$ is well defined, since it always holds that $\Y\to\K_{|Y|,\ptn(\Y)}$, as is witnessed by the identity homomorphism.
Furthermore, if $\Y$ is a loopless digraph, $\chi(\Y)$ is the standard (classical) chromatic number of $\Y$.

We shall prove the following intermediate result.

\begin{prop}
\label{prop_QU_coloring_separable_structures_general}
There exists a function $f:\N\to\N$ such that,
for any $d\in\N$ and any $\sigma$-structure $\Y$, 
it holds that $\Y^{\C^d}\to\K_{m,\ptn(\Y)}$
with $m=|Y|^{d}\cdot f(d)$.
\end{prop}

We first need to establish some geometric facts.
Let $S$ be a finite set, and take two measurements $\mathcal{E},\mathcal{F}\in\PVM_\HH(S)$. We say that $\mathcal{E}$ and $\mathcal{F}$ are \textit{commuting} if $[E_s,F_{\tilde s}]=O$ 
for each $s,\tilde s\in S$, where $E_s$ (resp. $F_{\tilde s}$) is the projector in $\mathcal{E}$ (resp. $\mathcal{F}$) corresponding to the outcome $s$ (resp. $\tilde s$). 
(Recall that $[E,F]$ is the commutator $EF-FE$.)

Given two commuting measurements $\mathcal{E}$ and $\mathcal{F}$, we can find a common basis of eigenvectors for the projectors appearing in them. In such a basis, there is an easy way to check whether the $s$-th projector of $\mathcal{E}$ is equal to the $s$-th projector of $\mathcal{F}$: Simply evaluate both of them onto all vectors of the basis and check if the outcomes always agree. The next lemma argues that, in fact, any basis \textit{close enough} to being an eigenbasis also allows for exact verification of whether the projectors coincide.

\begin{lem}
\label{lem_close_PVMs_means_equal}
    Fix a finite set $S$, an integer $d\in\N$, and two commuting measurements $\mathcal{E},\mathcal{F}\in\PVM_{\C^d}(S)$. The following statements are equivalent:
    \begin{enumerate}
        \item[$(i)$] $\mathcal{E}=\mathcal{F}$;
        \item[$(ii)$] there exists an orthonormal basis $\{\bv_1,\dots,\bv_d\}$ of $\C^d$ and a partition $\tau=\{\tau_s:s\in S\}$ of $[d]$ such that, for each $s\in S$ and each $j\in\tau_s$, it holds that $\norm{E_s\bv_j}>\alpha$ and $\norm{F_s\bv_j}>\alpha$ with $\alpha=\sqrt{\frac{1}{2}(1+\sqrt{1-1/d^2})}$.
    \end{enumerate}
\end{lem}
\begin{proof}
The implication $(i)\Rightarrow(ii)$ is obvious---just take a common eigenbasis for the projectors in $\mathcal{E}=\mathcal{F}$.
To prove the converse implication, 
%
take an orthonormal basis $\{\bv_1,\dots,\bv_d\}$ of $\C^d$ and a partition $\tau=\{\tau_s:s\in S\}$ of $[d]$ satisfying the condition of part $(ii)$ of the statement.
For $s\in S$ and $j\in\tau_s$, consider the vector $\bw_{s,j}=(I-F_s)\bv_j$. Observe that the vectors $F_s\bv_j$ and $\bw_{s,j}$ are orthogonal. Therefore,
\begin{align*}
    1&=
    \norm{\bv_j}^2
    =\norm{F_s\bv_j+\bw_{s,j}}^2\\
    &=\norm{F_s\bv_j}^2+\norm{\bw_{s,j}}^2
    > \alpha^2+\norm{\bw_{s,j}}^2,
\end{align*}
whence we deduce that
    $\norm{\bw_{s,j}}<\sqrt{1-\alpha^2}$. 
Define the operator $G_s=E_sF_s$.  %
We obtain
    \begin{align}
    \label{eqn_1307_2001}
    \notag
        \norm{G_s\bv_j}
        &=
        \norm{E_sF_s\bv_j}
        =
        \norm{E_s(\bv_j-\bw_{s,j})}\\
        \notag
        &\geq \norm{E_s\bv_j}-\norm{E_s\bw_{s,j}}
        >\alpha-\norm{\bw_{s,j}}\\
        &>\alpha-\sqrt{1-\alpha^2}
        =
        \sqrt{\frac{d-1}{d}},
    \end{align}
    which holds for each $s\in S$ and each $j\in\tau_s$.
 (The first inequality is the reverse triangle inequality, and the last equality follows from elementary algebra.)

Note now that each $G_s$ is a projector onto $\C^d$ as $[E_s,F_s]=O$. Also, for $s\neq\tilde s\in S$, we find 
\begin{align*}
G_s G_{\tilde s}=E_sF_sE_{\tilde s}F_{\tilde s}
=
E_sE_{\tilde s}F_{s}F_{\tilde s}
=
O,
\end{align*}
where we have used that the measurements $\mathcal{E}$ and $\mathcal{F}$ are commuting.
Hence, the operator $G=\sum_{s\in S}G_s$ is the sum of mutually orthogonal projectors, and it is thus a projector itself.
%
We find
\begin{align*}
    \tr(G)
    &=
    \sum_{j\in[d]}\bv_j^\ast G\bv_j
    =
    \sum_{s\in S}\sum_{j\in[d]}\bv_j^\ast G_s\bv_j\\
    &=\sum_{s\in S}\sum_{j\in[d]}\norm{G_s\bv_j}^2
    \geq 
    \sum_{s\in S}\sum_{j\in\tau_s}\norm{G_s\bv_j}^2\\
    &
    >
    \sum_{s\in S}\sum_{j\in\tau_s}\frac{d-1}{d}
    =
    d\frac{d-1}{d}=d-1,
\end{align*}
where the second equality holds since each $G_s$ is a projector, the second inequality comes from~\cref{eqn_1307_2001}, and the fourth equality holds since the set $\{\tau_s:s\in S\}$ partitions $[d]$. As a consequence, $\tr(G)=d$. It follows that $G=I$ and, thus, $\mathcal{E}=\mathcal{F}$, as required. 
\end{proof}

The next step is to show that there exists a \textit{finite} set $\mathcal{B}$ of ``indicator'' orthonormal bases in $\C^d$ having the property that any measurement is
``almost diagonalised'' in at least one of those bases. To prove this, we use the compactness of the unitary group in $\C^{d\times d}$. 

We note that, at a high level, it is the finiteness of $\mathcal{B}$
that allows any perfect quantum strategy to be simulated via finitely many bits of classical communication. Indeed, Alice and Bob will be able to classically simulate their quantum strategy by sending from one to the other information on the eigenstructure of the measurement that they perform upon receiving the Referee's question. Since finitely many bases are sufficient to be ``close enough'' to diagonalise any measurement, the information can be encoded in a finite message.

\begin{lem}
\label{lem_compactness_Ud}
    For each $d\in\N$ and each $1>\alpha\in\R$ there exists a finite set $\mathcal{B}$ of orthonormal bases of $\C^d$ such that, given any finite set $S$ and any measurement $\mathcal{E}\in\PVM_{\C^d}(S)$, for at least one basis $\{\bv_1,\dots,\bv_d\}\in\mathcal{B}$ there exists a partition $\tau=\{\tau_s:s\in S\}$ of $[d]$ such that $\norm{E_s\bv_j}>\alpha$ for each $s\in S$ and each $j\in\tau_s$.
\end{lem}
\begin{proof}
    We can associate a unitary matrix $M$ to each PVM $\mathcal{E}$ onto $\C^d$, by picking an orthonormal basis diagonalising all projectors in $\mathcal{E}$ and viewing them as the columns of $M$. Thus, the result directly follows from the compactness of the unitary group $U(d)$.
\end{proof}

We will also need the next observation, whose proof is deferred to~\cref{app_omitted_proofs}, showing that the pattern of a structure is equal to that of its quantum version.

\begin{restatable}{prop}{proppatternpreservationunderquantum}
\label{prop_pattern_preservation_under_quantum}
    For any $\sigma$-structure $\Y$ and any Hilbert space $\HH$ it holds that $\ptn(\Y)=\ptn(\Y^\HH)$.
\end{restatable}
\begin{proof}[Proof of \cref{prop_QU_coloring_separable_structures_general}]
Let $\mathcal{B}$ be a finite set of orthonormal bases of $\C^d$ witnessing the truth of \cref{lem_compactness_Ud} applied to $d$ and $\alpha=\sqrt{\frac{1}{2}(1+\sqrt{1-1/d^2})}$. Pick a vertex $\mathcal E$ of $\Y^{\C^d}$; i.e., a measurement $\mathcal E\in\PVM_{\C^d}(Y)$.
Let $n=|Y|$.
We colour $\mathcal E$ with the pair $\vartheta(\mathcal E)=(B,\tau)$, where $B$ is a basis from $\mathcal B$ and $\tau$ is a partition of $[d]$ in $n$-many parts, satisfying the condition in \cref{lem_compactness_Ud}. 
Observe that there are at most $n^d$ such partitions $\tau$. Hence,
letting $f(d)$ be the size of $\mathcal{B}$, $\vartheta$ yields a function from the domain of $\Y^{\C^d}$ to a subset of $[f(d)]\times [n^d]$, which we identify with a subset of the vertex set of $\K_{m,\ptn(\Y)}$. 

We are left to show that $\vartheta$ is a proper colouring.
Take a symbol $R\in\sigma$ of some arity $r$, and a tuple $\be=(\mathcal{E}^{(1)},\dots,\mathcal{E}^{(r)})\in R^{\Y^{\C^d}}$. 
Using \cref{prop_pattern_preservation_under_quantum}, we deduce that  $\pi_\be$ belongs to the $R$-th entry of $\ptn(\Y^{\C^d})=\ptn(\Y)$; i.e., writing $\ptn(\Y)=(P_R)_{R\in\sigma}$, we have that $\pi_\be\in P_R$.
Let $\bc=(\vartheta(\mathcal{E}^{(1)}),\dots,\vartheta(\mathcal{E}^{(r)}))$; we will argue that $\pi_\bc = \pi_\be$, which implies that $\bc \in R^{\K_{m,\ptn(\Y)}}$ by the definition of the complete structure.
Suppose, for the sake of contradiction, that $\pi_\bc \neq \pi_\be$.
This means that there exist two indices $i,j\in [r]$ such that $\mathcal{E}^{(i)}\neq\mathcal{E}^{(j)}$ but
$\vartheta(\mathcal{E}^{(i)})=\vartheta(\mathcal{E}^{(j)})$. Let $(B,\tau)$ be the common colour of $\mathcal{E}^{(i)}$ and $\mathcal{E}^{(j)}$, with $B=(\bv_1,\dots,\bv_d)$. By construction of $\vartheta$, it must hold that $\norm{E_y^{(i)}\bv_j}>\alpha$ and $\norm{E_y^{(j)}\bv_j}>\alpha$ for each $y\in Y$ and each $j\in\tau_y$.

It easily follows from the definition of $\Y^\HH$ in \cref{subsec_quantum_strategies} that the two measurements $\mathcal{E}^{(i)}$ and $\mathcal{E}^{(j)}$ are commuting. Indeed, there exists some $\mathcal F\in\PVM_{\C^d}(R^{\Y})$ such that, for each $y,\tilde y\in Y$, we have
\begin{align*}
    E^{(i)}_y=\sum_{\substack{\by\in R^{\Y}\\y_i=y}}F_{\by},\quad
    E^{(j)}_{\tilde y}=\sum_{\substack{\by\in R^{\Y}\\y_j=\tilde y}}F_{\by},
\end{align*}
so the fact that $\mathcal{F}$ is a PVM implies that $[E^{(i)}_y,E^{(j)}_{\tilde y}]=O$. We can then apply \cref{lem_close_PVMs_means_equal} to deduce that $\mathcal{E}^{(i)}=\mathcal{E}^{(j)}$, a contradiction.
\end{proof}
The next result, whose proof is deferred to~\cref{app_omitted_proofs}, shows that, just like in the binary case, the structure $\uprelstr {k,\Y}$ contains arbitrarily large cliques when $k$ grows large. Upper-bounding the dependence of $k$ on the size of the complete structure will require some additional care, though.

\begin{restatable}{prop}{propanycliquesinAlice}
\label{prop_any_cliques_in_Alice}
    Take a $\sigma$-structure $\Y$ and an integer $n\in\N$, and let $\ptn(\Y)=(P_R)_{R\in\sigma}$.
    Then
    $
        \K_{n,\ptn(\Y)}\to\uprelstr {k,\Y}
    $
    with $k = \sum_{R \in \sigma} |P_R| \cdot \mathcal{O}_{\ar(R)}(\log_2 n)$.
\end{restatable}

We can now finalise the proof of our main results.

\begin{proof}[Proof of~\cref{thm_main_alice_general}]
     For $n=\chi(\Y)$, it holds that $\Y\to\K_{n,\ptn(\Y)}$. Using~\cref{YH_monadic},~\cref{prop_QU_coloring_separable_structures_general}, and~\cref{prop_any_cliques_in_Alice}, we obtain
    \begin{align*}
        \Y^\HH\to\K_{n,\ptn(\Y)}^\HH\to\K_{m,\ptn(\Y)}\to\uprelstr {k,\Y}
    \end{align*}
    for suitable $m,k\in\N$.
    By virtue of~\cref{thm_quantum_csp_is_csp} and~\cref{thm_algebraic_characterisation_alice_to_bob}, the above homomorphism is equivalent to the claim.
\end{proof}

The proof of~\cref{thm_main_bob_general} requires a few more steps. Indeed, the last link of the chain $\Y^\HH\to\K_{n,\ptn(\Y)}^\HH\to\K_{m,\ptn(\Y)}\to \uprelstrbob{k,\Y}$ does not hold in this case, as arbitrarily large complete structures can be found in Bob powers of only certain structures $\Y$, unlike for Alice powers.\footnote{\label{footnote_large_bob_cliques}For the reader's benefit, we have included in~\cref{appendix_large_cliques_in_bob_powers} a characterisation of when this happens.}
Hence, we construct the desired homomorphism $\Y^\HH \to \uprelstrbob{k, \Y}$ directly.

\begin{proof}[Proof of~\cref{thm_main_bob_general}]
    By~\cref{thm_quantum_csp_is_csp} and~\cref{thm_algebraic_characterisation_bob_to_alice}, the claim is equivalent to the fact that $\Y^\HH\to\uprelstrbob{k,\Y}$.

    \cref{prop_QU_coloring_separable_structures_general} yields a homomorphism $c$ from $\Y^\HH$ to $\K_{m,\ptn(\Y)}$ for some $m$.
Moreover, it is not hard to check that $c$ is pattern-preserving, in the sense that $\pi_\be = \pi_{c(\be)}$ for each symbol $R \in \sigma$ and each tuple $\be \in R^{\Y^\HH}$. Define a non-zero selector function $s : \PVM_\HH(Y) \to Y$ which, for any $\mathcal{E} \in \PVM_\HH(Y)$, selects an arbitrary element $s(\mathcal{E})$ with $E_{s(\mathcal{E})} \neq O$.
Let $k = m |Y|$, and consider the function
$h : \PVM_\HH(Y) \to [k] \times Y$ defined by
    $$\PVM_\HH(Y) \ni \mathcal{E} \mapsto ((c(\mathcal{E}), s(\mathcal{E})), \, s(\mathcal{E})) \in [k] \times Y.$$
    We claim that $h$ is the required homomorphism from $\Y^\HH$ to $\uprelstrbob{k,\Y}$.

    Fix any $R \in \sigma$ of arity, say, $r$; pick a tuple $\be =(\mathcal{E}^{(1)}, \ldots, \mathcal{E}^{(r)}) \in R^{\Y^\HH}$ and consider the entrywise application of the function $h$.
    To confirm that the tuple $h(\be)$ is in $R^{\uprelstrbob{k, \Y}}$, pick a pair $(\tilde c, \tilde s) \in [m]\times Y$ and define the set $S = \{i \in [r] : c(\mathcal{E}^{(i)}) = \tilde c \text{ and } s(\mathcal{E}^{(i)}) = \tilde s\}$.
    We claim that some tuple in $R^\Y$ agrees with $s(\be)$ on all coordinates in $S$. The case $S = \emptyset$ is trivial.
    Otherwise, we proceed
    to find an element $\mathcal{F} \in \PVM_\HH(R^\Y)$ for which $E^{(i)}_{\tilde{s}} = \sum_{\by \in R^\Y : y_i = \tilde{s}} F_\by$ for each $i \in S$, as per~\cref{eqn_defining_quantum_structure}.
    %
    Recall that the projectors in $\mathcal{F}$ are pairwise orthogonal.
    Furthermore, using that $\pi_\be=\pi_{c(\be)}$, we obtain that $E^{(i)}_{\tilde s}$ is the same non-zero projector for each $i\in S$.
    Therefore, there must exist some $\tilde\by\in R^\Y$ such that $F_{\tilde{\by}}$ appears in the summation of each of them. Hence, $\tilde{\by}$ is such that $\tilde{y}_i = \tilde{s} = s(\mathcal{E}^{(i)})$ for all $i \in S$ and, thus, it witnesses that $h(\be)\in R^{\uprelstrbob{k,\Y}}$.
\end{proof}

\begin{rem}
For clarity, 
the statements of \cref{thm_main_alice_general} and \cref{thm_main_bob_general} do not specify any upper bound on $k$.
However, bounds can be obtained as follows.
Let $f$ be the function from \cref{prop_QU_coloring_separable_structures_general} and $d = \dim\HH$.
Then,\footnote{By $\mathcal{O}_{\ptn(\Y)}$ we mean that the hidden constant can depend on $|P_R|$ and $\ar(R)$ for all $R \in \sigma$.} the value of $k$ in~\cref{thm_main_alice_general} is $\mathcal{O}_{\ptn(\Y)}(d\log \chi(\Y) + \log f(d))$,
while the value of $k$ in~\cref{thm_main_bob_general} is $\leq|Y|^{d+1} \cdot f(d)$.
\end{rem}

\subsection{The quantum vs. classical chromatic gap}
\label{subsec_quantum_vs_classical_chromatic_gap}
Given a graph $\X$ and an integer $d\in\N$, let $\chi_d(\X)$ be the \textit{quantum chromatic number} of $\X$ in dimension $d$; i.e., the minimum $n\in\N$ such that $\X\to\K_n^{\C^d}$. 
The results proved in~\cref{sec_entangled_vs_classical} can be used to bound the gap between quantum and classical chromatic number of graphs in any fixed dimension $d$.

Since $\K_n\to\K_n^{\C^d}$, it is clear that $\chi_d(\X)\leq\chi(\X)$ for each $\X$.
For $d\geq 3$, 
the following lower bound on the gap between $\chi_d$ and $\chi$ was established in~\cite[Cor.35--Ex.36]{ciardo_quantum_minion}.
\begin{thm}\cite{ciardo_quantum_minion}
    For each $d\geq 3$ and each $n\geq 3$ there exists a graph $\X$ such that $\chi_d(\X)\leq n$ but $\chi(\X)\geq 2n-1$.
\end{thm}
To the best of the authors' knowledge, prior to the current work, it was an open question whether the $\chi_d$ vs. $\chi$ gap is bounded or not. Using~\cref{prop_QKn_to_Km}, we answer in the affirmative. 
\begin{thm}
\label{thm_upper_bound_qcn}
For each graph $\X$ and each $d\in\N$, it holds that
$
    \chi(\X)\leq \alpha_d\cdot\chi_d(\X),
$
where $\alpha_d=(2+\mathcal{o}(1))^{d-1}$.
\end{thm}
\begin{proof}
    Suppose that $\chi_d(\X)=n$. Using~\cref{prop_QKn_to_Km}, we obtain
$
\X\to\K_n^{\C^d}\to\K_{n\cdot\alpha_d}$,
    whence the conclusion follows.
\end{proof}
We point out that~\cref{thm_upper_bound_qcn} implies that the quantum vs. classical chromatic gap is bounded when the dimension $d$ of the Hilbert space describing the system measured in the quantum protocol is \textit{fixed}. On the other hand, 
since $\alpha_d\to\infty$ when $d\to\infty$, the theorem does not imply that the gap is bounded when the dimension of the Hilbert space is arbitrary.
In fact, it was recently proved in~\cite{quantum_chromatic_gap} that, conditional to the existence of strongly pseudo-telepathic label-cover or $3$XOR instances, there exist graphs $\X$ such that $\chi_d(\X)=3$ for some $d$, but $\chi(\X)$ is arbitrarily large.\footnote{When $d\leq 2$, it follows from~\cite[Cor.13]{ciardo_quantum_minion} that $\Y^{\C^d}$ is homomorphically equivalent to $\Y$ for each structure $\Y$, so $\chi_2(\X)=\chi_1(\X)=\chi(\X)$ for every graph $\X$.}

We also point out that our definition of $\chi_d$ relies on the notion of quantum strategies as described in~\cref{subsec_quantum_strategies}, which is standard in the literature on quantum CSPs (see e.g.~\cite{abramsky2017quantum}). In particular, for such strategies, measurements corresponding to adjacent vertices must commute. In the literature on the quantum chromatic number, the commutativity assumption is usually not enforced (\cite{CameronMNSW07,MancinskaR16_oddities,MancinskaR16}; see also~\cite{paulsen2015quantum} for a presentation of several variants of the quantum chromatic number). The non-commutative version of quantum chromatic number in dimension $d$ is clearly not larger than $\chi_d$.\footnote{In fact, there are examples of graphs for which the two versions differ~\cite{amin_karamlou_quantum_chromatic_number}.}

\section{Further directions}
In the current paper, we explored the connection between classical strategies for CSPs with extra communication channels and those assisted by shared quantum entanglement. Another direction for further research appears to be understanding the complexity of the computational problem associated with the former type of strategies: 
\begin{center}
\textit{How hard is it to check whether a perfect classical strategy with extra communication exists?} 
\end{center}
For concreteness, we now discuss the case of strategies of the type ``Alice messages Bob''. The same questions can be asked for ``Bob messages Alice'' strategies.

Following the algebraic approach to CSPs, the problem consists in understanding the structure of the polymorphism clone $\Pol(\uprelstr{k,\Y})$, and its connection with $\Pol(\Y)$.\footnote{For the technical definitions, we refer the reader to~\cite{BKW17}.} Perhaps surprisingly, whatever connection exists between the two clones, it is not going to preserve clone (or minion\footnote{Minion homomorphisms are sometimes referred to as \textit{height-1 clone homomorphisms} in the literature, see e.g.~\cite{BKW17,BOP18}.}) homomorphism, as the next example illustrates.

\begin{example}\label{ex_Alice_power_of_hard_csp_is_p} 
Let $R_a$, for $a \in \{0,1\}$, denote the relation $\{0,1\}^3 \setminus\{(a,a,a)\}$.
    Consider the following pair of Boolean structures: $\Y$ has two ternary relations, $R_0$ and $R_1$, while $\Y'$ is the extension of $\Y$ with the additional ternary relation from $\NAE_2$.
    It is known that $\Pol(\Y) = \Pol(\Y')$, and the corresponding CSPs are both NP-complete \cite{BKW17}.

    Now, consider the Alice squares of both structures. \cref{example_AliceToBob_cliques} asserts that $\uprelstr{2,\NAE_2} = \NAE_4$, which defines an NP-complete CSP.
    In particular, this means that $\CSP(\uprelstr{2,\Y'})$ is NP-complete.
    On the other hand, both relations of $\uprelstr{2,\Y}$ contain a constant triple of $(0,1)$, which makes its $\CSP$ trivial.
    This implies that no clone (or minion) homomorphism can exist from $\Pol(\uprelstr{2,\Y})$ to $\Pol(\uprelstr{2,\Y'})$, because it would provide a reduction from an NP-complete CSP to a tractable one \cite{BOP18}.
    This example also illustrates that, while for many structures it happens that $\CSP(\Y)$ reduces to $\CSP(\uprelstr{k,\Y})$ via \textit{gadget} reductions, there exist NP-complete structures whose Alice power is tractable.
\end{example}

As a consequence, methods more advanced than clone or minion homomorphisms are needed to answer the above question.
A related direction is to study the complexity of the \textit{promise problem} of distinguishing, for an instance $\X$, whether $\X\in\CSP(\Y)$ or $\X\not\in\GameAB{k,\Y}$.
This is the \textit{promise CSP} associated with the template $(\Y,\uprelstr{k,\Y})$, in symbols $\PCSP(\Y,\uprelstr{k,\Y})$~\cite{AGH17,BG21:sicomp,BBKO21}.
\begin{example}
    It is not difficult to check that $\PCSP(\Y,\uprelstr{k,\Y})$ is not always tractable.
    For example, if $\Y$ is the clique $\K_n$, we have from~\cref{example_AliceToBob_cliques} that $\PCSP(\K_n,\uprelstr{k,\K_n})=\PCSP(\K_n,\K_{n^k})$, which is known to be NP-complete for $n$ large enough~\cite{KOWZ22}. 
    
    On the other hand, there are cases in which both $\Y$ and $\uprelstr{k,\Y}$ are NP-complete, but the promise problem is not. For example, consider the NP-complete structure $\textbf{1-in-3}$ (i.e., the Boolean structure having a single, ternary relation $\{(1,0,0),(0,1,0),(0,0,1)\}$). It is not hard to check that $\uprelstr{2,\textbf{1-in-3}}$ contains a copy of $\NAE_2$ and it is NP-complete. However, the chain
$
        \textbf{1-in-3}
        \to\NAE_2
        \to
        \uprelstr{2,\textbf{1-in-3}}
$
    implies that $\PCSP(\textbf{1-in-3},\uprelstr{2,\textbf{1-in-3}})$ reduces to $\PCSP(\textbf{1-in-3},\NAE_2)$, which is tractable in polynomial time~\cite{BG21:sicomp}.
\end{example}

\bibliographystyle{alpha}
\bibliography{bibliography.bib}

\appendix

\section{Large complete structures in Bob powers}
\label[appendix]{appendix_large_cliques_in_bob_powers}

This section is devoted to characterising structures $\Y$ that contain arbitrarily large complete structures in their (sufficiently large) Bob power, as mentioned in~\cref{footnote_large_bob_cliques}. We will capture such structures through the following definition.
\begin{defn}
    Let $\Y$ be a $\sigma$-structure and write $\ptn(\Y)=(P_R)_{R\in\sigma}$. A vertex $v \in Y$ is called \emph{central in $\Y$} if, for each $R\in\sigma$ and for each $V\in\pi \in P_R$, there is some tuple $\by\in R^\Y$ such that $y_p=v$ for each $p\in V$.
\end{defn}

\begin{example}
    Consider the structure $\Y$ with domain $Y = [7]$ and a single 4-ary relation $R^\Y = \{(1, 1, 2, 3), (4, 5, 1, 1), (6, 6, 7, 7)\}$. The vertex $1$ is central, but it does not appear in any constraint with the maximal pattern $\{\{1, 2\}, \{3, 4\}\}\in P_R$.
\end{example}

\begin{prop}
    For any $\sigma$-structure $\Y$ and any $n\in\N$, the following statements are equivalent:
    \begin{enumerate}
        \item[$(i)$] there exists some $k\in\N$ such that $\K_{n,\ptn(\Y)}\to\uprelstrbob{k,\Y}$;
        \item[$(ii)$] there exists $v \in Y$ that is central in $\Y$. 
    \end{enumerate}
\end{prop}

\begin{proof}
    For the direction $(i)$ $\implies$ $(ii)$, pick $k \in \N$ such that there exists a homomorphism $h$ from $\K_{n,\ptn(\Y)}$ to $\uprelstrbob{k,\Y}$.
    Let $h(1) = (s, v) \in [k] \times Y$.
    We will argue that $v$ is central in $\Y$.
    Fix a relation symbol $R \in \sigma$, and a set $V \in \pi \in P_R$.
    By the definition of $\K_{n,\ptn(\Y)}$, there exists a tuple $\ba \in R^{\K_{n,\ptn(\Y)}}$ with $a_i = 1$ for all $i \in V$ and $\pi_\ba = \pi$.
    Applying the homomorphism $h$ to $\ba$ entrywise, we deduce that there exists $\tilde{\by} \in R^{\Y}$ such that $\tilde{\by}_i = v$ for all $i \in V$, as required.

    For the opposite direction, fix a central element $v \in Y$. We claim that $k = n$ is sufficient and the homomorphism is
    \begin{equation*}
        h : [n] \ni i \mapsto (i, v) \in [n] \times Y
    \end{equation*}
    Fix a symbol $R \in \sigma$ and a tuple $\ba \in R^{\K_{n,\ptn(\Y)}}$.
    Write $\textbf{t} = (h(a_1), \ldots, h(a_{\ar(R)}))$, and observe that $ \pi_{\textbf{t}} = \pi_\ba$.
    We will argue that $\mathbf t \in R^{\uprelstrbob{n,\Y}}$.
    For each $s \in [n]$, either the set $S = \{i \in [\ar(R)] : h(a_i) = (s, v)\}$ is empty, or it is equal to some $V \in \pi_\ba$. The former case is trivial. In the latter case, by the definition of a central vertex, there exists some $\tilde{\by} \in R^\Y$ with $\tilde{y}_i = v$ for each $i \in V$. Hence, $\tilde{\by}$ is the required witness for $s$.
\end{proof}

\section{Omitted proofs}
\label[appendix]{app_omitted_proofs}
In this section, we present the proofs that have been omitted from the main body of the paper.

\thmalgebraiccharacterisationbobtoalice*
\begin{proof}
    We begin by proving the right-to-left inclusion.
    Take a homomorphism $h:\X\to\uprelstrbob{k, \Y}$ for some relational structure $\X$. In the strategising phase for $\GameBA{k, Y}$, Alice and Bob agree on fixing one such homomorphism $h$.
    In the game phase, upon receiving a vertex $x \in X$ from the Referee, Bob computes $h(x) = (s, y)$, where $s\in[k]$ and $y\in Y$.
    He sends $y$ to the Referee and $s$ to Alice.
    Alice receives $\bx \in R^\X$ from the Referee and $s$ from Bob.
    She then computes the set 
    \begin{align*}
    S = \{i \in [r] : \exists{z_i\in Y}~h(x_i) = (s, z_i)\}\text{,} 
    \end{align*}
    where $r$ denotes arity of $R$. From the fact that $h$ is a homomorphism and by the definition of $\uprelstrbob{k,\Y}$ applied to $h(\bx) \in R^{\uprelstrbob{k,\Y}}$, we deduce that there exists some tuple $\by = (y_1, \ldots, y_r) \in R^\relstr Y$ consistent with $h(\bx)$ on $S$ (in the sense that for $i \in S$, we have that $h(x_i) = (s, y_i)$), which Alice sends to the Referee.
    Since Bob's vertex $x$ can only be present in $\bx$ on indices from $S$, their responses are consistent and they win the game.    

    For the other inclusion: fix $\X$ such that Alice and Bob can win $\GameBA{k, Y}$.
    This means that there exists a perfect strategy, which is encoded by the functions $h: X\to [k]\times Y$ given by $h(x) = (m(x), b(x))$ and $a_R:[k]\times R^\X\to R^\Y$ for each symbol $R$. Here, whenever Bob receives the question $x \in X$ from the Referee, he will send Alice the message $m(x)\in[k]$ and will play $b(x) \in Y$; whenever Alice receives the message $s\in[k]$ from Bob and the question $\bx\in R^{\X}$ from the Referee, she will play $a_R(s,\bx)\in R^{\Y}$.
%

    We claim that $h$ establishes the required homomorphism from $\X$ to $\uprelstrbob{k,\Y}$. 
    Fix a symbol $R$ of some arity $r$ in the signature of $\X$ and $\Y$, and a tuple $\bx \in R^\X$. Given any $s\in [k]$, we aim to use Alice's reply $a_R(s,\bx)$ as a witness for the fact that $h(\bx)\in R^{\uprelstrbob{k,\Y}}$. Concretely, we need to prove that, for each $j\in[r]$ for which $m(x_j)=s$, it holds that $b(x_j)=a_R(s,\bx)_j$. But this is precisely the winning condition characterising the existence of a perfect strategy, and it is thus verified by assumption.
\end{proof}

\thmwhenbithelpsbob*
\noindent We will need the following simple lemma.
\begin{lem}
\label{lem_basic_bob}
    For every $k$, if $\relstr X \rightarrow \relstr Y$ then $\uprelstrbob{k,\relstr X}\rightarrow \uprelstrbob{k,\relstr Y}$.
    In particular $\uprelstrbob{k, \relstr Y}$ and $\uprelstrbob{k, \core{\relstr Y}}$
    are homomorphically equivalent (and, thus, define the same CSP).
\end{lem}
\begin{proof}
    For the first claim, the map $(s,x) \mapsto (s,h(x))$ is a required homomorphism.
    The second claim is then immediate.
\end{proof}
\begin{proof}[Proof of \cref{thm:whenbithelpsbob}]
Note that the condition $\GameBA{2, \relstr Y}=\CSP(\Y)$ is equivalent to $\uprelstrbob {2,\Y}\to\relstr Y$, since it always holds that $\Y\to\uprelstrbob {2,\Y}$. Moreover, by~\cref{lem_basic_bob}, we can substitute $\Y$ with $\core\Y$, and the games do not change.

    The implication from right to left is straightforward: If every relation in $\Y$ is a Cartesian product, then let $h:[2]\times Y\to Y$ be defined as $h(s, y) = y$. To see that $h$ is indeed a homomorphism from $\uprelstrbob{2,\Y}$ to $\Y$, consider any symbol $R \in \sigma$ and any tuple $\ba \in R^{\uprelstrbob{2,\Y}}$. Since $R^\Y$ is a Cartesian product, the tuple $h(\ba)$ also belongs to it.

    For the left-to-right implication, fix a core $\Y$ and a homomorphism $h:\uprelstrbob{2,\Y}\to\Y$.
    For each $s \in [2]$, let $h_s: Y \ni y \mapsto h(s, y) \in Y$. Clearly, both $h_s$ are endomorphisms of $\Y$. Since $\Y$ is a core, they are in fact automorphisms of $\Y$.

    Now, fix any $R \in \sigma$, say of arity $r$.
    We will denote by $\pi_{1..i} R^\Y$ the projection of the relation onto coordinates $(1, \dots, i)$, i.e. $\pi_{1..i} R^\Y = \{(y_1, \dots, y_i) \in Y^i \mid \exists{y_{i+1}, \dots, y_r}~(y_1, \dots, y_r) \in R^\Y\}$.
    Moreover, we define $\pi_i R^\Y$ in a similar manner.
    Fix any tuple $\by \in \pi_1 R^\Y \times \dots \times \pi_{r} R^\Y$; we claim that $\by \in R^\Y$.
    Since both $h_s$ are permutations, there exists an integer $M$ such that $h_1^M = h_2^M = id$.

    We will argue using induction on $i \in [r]$ that $(y_1, \dots, y_i) \in \pi_{1..i} R^\Y$.
    The base case $i = 1$ holds by the definition of $\by$.
    For the inductive step, suppose that the claim holds for some $i \in [r]$.
    Observe that $h_1^{M-1}(y_1, \dots, y_i) \in \pi_{1..i} R^\Y$ and $h_2^{M-1}(y_{i+1}) \in \pi_{i+1} R^\Y$.
    By the definition of $\uprelstrbob{2,\Y}$, we have
    \begin{equation*}
    \Big( (1, h_1^{M-1}(y_1)), \dots, (1, h_{1}^{M-1}(y_i)), (2, h_2^{M-1}(y_{i+1})) \Big) \in R^{\uprelstrbob{2,\Y}}.
    \end{equation*}
    Therefore, 
    \begin{equation*}
    \left( h_1^M(y_1), \dots, h_{1}^M(y_{i}), h_2^M(y_{i+1}) \right) = (y_1, \dots, y_{i+1}) \in \pi_{1..i+1} R^\Y
    \end{equation*}
    since $h$ is a homomorphism.
    Hence, the claim is true by induction.
    It follows that $\by \in \pi_{1..r} R^\Y = R^\Y$.
    We proved that $R^\Y$ is a Cartesian product.
\end{proof}

\proppatternpreservationunderquantum*
\begin{proof}
Let $\sigma$ be the signature of $\Y$, and write $\mathcal{P}=\ptn(\Y)=(P_R)_{R\in\sigma}$ and $\mathcal{Q}=\ptn(\Y^\HH)=(Q_R)_{R\in\sigma}$. Fix a symbol $R\in\sigma$ of some arity $r$. We need to show that $P_R=Q_R$. 

For the left-to-right inclusion, pick some partition $\pi\in P_R$, and take $\by\in R^\Y$ such that $\pi\preceq\pi_\by$. Take the tuple $\be=(\mathcal{E}^{(1)},\dots,\mathcal{E}^{(r)})$ of elements in $\PVM_\HH(Y)$ defined by 
\begin{align*}
    E^{(i)}_y=\left\{\begin{array}{lll}
         I&\mbox{if }y_i=y  \\
         \{\bzero\}&\mbox{otherwise}.
    \end{array}\right.
\end{align*}
 Consider the measurement $\mathcal{F}\in\PVM_{\HH}(R^\Y)$ defined by setting $F_\by=I$, $F_\bz=O$ for each $\bz\neq\by$. We see from the definition of $\Y^\HH$ that $\mathcal{F}$ witnesses the fact that $\be\in R^{\Y^\HH}$ (recall the discussion in \cref{subsec_quantum_strategies}). 
 Moreover, it is straightforward to check that $\pi_\be=\pi_\by$. We deduce that $\pi\preceq\pi_\be\in Q_R$, thus proving that $P_R\subseteq Q_R$.

 For the other inclusion, start with a partition $\pi\in Q_R$, and let this be witnessed by the fact that $\pi\preceq\pi_\be$ for some $\be=(\mathcal{E}^{(1)},\dots,\mathcal{E}^{(r)})\in R^{\Y^\HH}$. Recall from~\cref{eqn_defining_quantum_structure} that there exists some $\mathcal{F}\in\PVM_\HH(R^\Y)$ for which the identity
 \begin{align*}
     E^{(i)}_{y}=\sum_{\substack{\by\in R^\Y\\y_i=y}}F_\by
 \end{align*}
 holds for each $i\in[r]$ and each $y\in Y$. Choose $\by\in R^\Y$ such that $\dim(F_\by)>0$. Since the projectors in $\mathcal{F}$ are mutually orthogonal, if $y_i\neq y_j$ for some $i,j\in[r]$, the projectors $E^{(i)}$ and $E^{(j)}$ must be different. This means that $\pi_\be\preceq\pi_\by$ for each $\by\in R^\Y$ such that $F_\by$ is nontrivial. Since $\dim(\HH)>0$, at least one such $\by$ exists. Hence, the chain $\pi\preceq\pi_\be\preceq\pi_\by\in P_R$ holds, which means that $Q_R\subseteq P_R$, thus concluding the proof.
\end{proof}

\propanycliquesinAlice*
\noindent We will utilise the following combinatorial result.
\begin{thm}[\cite{godbole1996covarray}]\label{thm_covering_matrix}
    For any $n, r\in\N$ and any finite set $X$ there exists an $n \times m$ matrix over $X$ such that
    \begin{itemize}
        \item $m = \Theta_{r, |X|}(\log_2 n)$, and
        \item any subset of rows $S \subseteq [n]$ of size $r$ contains, among its columns, all possible tuples in $X^r$.
    \end{itemize}
\end{thm}
\begin{proof}[Proof of \cref{prop_any_cliques_in_Alice}]
    For each $R \in \sigma$ and each $\pi \in P_R$, let $M_{R, \pi}$ be the matrix obtained from \cref{thm_covering_matrix} applied to $n, \ar(R)$, and $X \subseteq Y$, where $X$ is defined as follows:
    We associate with the matrix $M_{R,\pi}$ an arbitrary tuple $\by \in R^\Y$ such that $\pi \preceq \pi_\by$, and let $X = \cup_i \{y_i\}$.
    Note that $|X| \le \ar(R) = r$, so $M_{R,\pi}$ has $\mathcal{O}_{r}(\log_2 n)$ columns.
    We claim that the rows of the matrix, viewed as vertices of the structure $\uprelstr{k,(X; \{\by\})}$, induce a copy of $\K_{n,(\{\pi_\by\})}$.

    Indeed, fix any $r$-ary tuple $\ba$ of the rows of $M_{R,\pi}$ such that $\pi_\ba \preceq \pi$.
    By the definition of $M_{R,\pi}$, among the columns of $\ba$ we can find any tuple in $X^r$ as long as its pattern is at most as fine as $\pi_\ba$.
    In particular, one of the columns is exactly $\by$, so it witnesses $\ba \in \uprelstr{k,(X; \{\by\})}$.
    
    Now, let $M$ be the matrix obtained by concatenating all $M_{R, \pi}$'s.
    Note that the number of columns is $k = \sum_{R \in \sigma} |P_R| \cdot \mathcal{O}_{\ar(R)}(\log_2 n)$.
    Then, view the rows of the matrix as the vertices of $\uprelstr{k,\Y}$.
    We will argue that these vertices induce $\K_{n,\ptn(\Y)}$ in $\uprelstr{k,\Y}$.

    Fix a symbol $R \in \sigma$ of arity, say, $r$.
    Fix any $r$-ary tuple $\ba$ of the rows of $M$ such that $\pi_\ba \in P_R$.
    Recall that, in this case, $\ba \in \uprelstr{k,\Y}$ just by looking at the columns originating from $M_{R,\pi_\ba}$.
    Hence, the proof follows.
\end{proof}

\end{document}